%% file: backward-subgraph.tex
\documentclass{llncs}

\usepackage{etex}
\usepackage[long]{optional}
\usepackage{amssymb,amsmath,latexsym,stmaryrd}
\usepackage{graphicx,color}
\usepackage{subfig}
\usepackage{enumerate}

\input{tikz-defs}

\input{general-defs}

\newcommand{\genArrow}{\tikzArrowToText{mor-genorder}}
\newcommand{\subArrow}{\tikzArrowToText{mor-subgraph}}
\newcommand{\indsubArrow}{\mathop{\textnormal{\texttt{$\triangleright$}}\;\!\!\!\!\to}}
\newcommand{\genOrder}{\sqsubseteq}
\newcommand{\subOrder}{\subseteq}
\newcommand{\indsubOrder}{\unlhd}
\newcommand{\predBasis}[2][]{\ensuremath{\mathit{pb}\ifthenelse{\equal{#1}{}}{}{_{#1}}(#2)}}
\newcommand{\pred}[2][]{\ensuremath{\mathit{Pred}\ifthenelse{\equal{#1}{}}{}{_{#1}}(#2)}}
\newcommand{\predAll}[1]{\ensuremath{\mathit{Pred}^*(#1)}}
\newcommand{\arity}{\mathit{ar}}
\newcommand{\ignore}[1]{}
\newcommand{\pseudoParagraph}[1]{\noindent\textit{#1:} }
\newcommand{\minskyDec}[1]{#1\mathord{--}}
\newcommand{\minskyInc}[1]{#1\mathord{++}}
\newcommand{\minskyZero}[1]{#1\mathord{=0?}}

\spnewtheorem{procedure}{Procedure}{\bfseries}{\rmfamily}

\newcommand*{\graphDiagramScale}{0.75}
\newcommand*{\ExampleRuleScale}{0.75}
\newcommand*{\catDiagramScale}{0.85}
\tikzset{StdGraphGrid/.style={x=2cm, y=-1.2cm}}

\title{A General Framework for Well-Structured Graph Transformation 
Systems\thanks{Research partially supported by DFG project GaReV.}}

\author{Barbara K\"onig \and Jan St\"uckrath}

\institute{%
  Universit\"at~Duisburg-Essen, Germany\\
  \email{\{barbara\_koenig, jan.stueckrath\}@uni-due.de}}
  
\begin{document}

\maketitle

\begin{abstract}
  Graph transformation systems (GTSs) can be seen as well-structured
  transition systems (WSTSs), thus obtaining decidability results for
  certain classes of GTSs. In earlier work it was shown that
  well-structuredness can be obtained using the minor ordering as a
  well-quasi-order. In this paper we extend this idea to obtain a
  general framework in which several types of GTSs can be seen as
  (restricted) WSTSs. We instantiate this framework with the subgraph
  ordering and the induced subgraph ordering and apply it to analyse a
  simple access rights management system.
\end{abstract}

\section{Introduction}
\label{sec:introduction}

Well-structured transition systems
\cite{acjt:general-decidability,fs:well-structured-everywhere} are one
of the main sources for decidability results for infinite-state
systems. They equip a state space with a quasi-order, which must be
a well-quasi-order (wqo) and a simulation relation for the transition
relation. If a system can be seen as a WSTS, one can decide the
coverability problem, i.e., the problem of verifying whether, from a
given initial state, one can reach a state that covers a final state,
i.e., is larger than the final state with respect to the chosen order.
Often, these given final states, and all larger states, are considered
to be error states and one can hence check whether an error state is
reachable.
Large classes of infinite-state systems are well-structured, for
instance (unbounded) Petri nets and certain lossy systems. For these
classes of systems the theory provides a generic backwards
reachability algorithm.

A natural specification language for concurrent, distributed systems
with a variable topology are graph transformation systems
\cite{r:gra-handbook} and they usually generate infinite state spaces.
%(even if one factors the state space through graph isomorphism). 
In those systems states are represented by graphs and state changes by
(local) transformation rules, consisting of a left-hand and a
right-hand side graph.
In \cite{JK08} it was shown how lossy GTSs with edge contraction rules
can be viewed as WSTSs with the graph minor ordering
\cite{rs:graph-minors-xx,RS:graphMinors:XXIII} and the theory was
applied to verify a leader election protocol and a termination
detection protocol \cite{bdkss:undecidability-gts}. The technique
works for arbitrary (hyper-)graphs, i.e.\ the state space is not
restricted to certain types of graphs. On the other hand, in order to
obtain well-structuredness, we can only allow certain rule sets, for
instance one has to require an edge contraction rule for each edge
label.

In order to make the framework more flexible we now consider other
wqos, different from the minor ordering: the subgraph ordering and the
induced subgraph ordering. The subgraph ordering and a corresponding
WSTS were already studied in \cite{BKWZ:2013}, but without the
backwards search algorithm. Furthermore, we already mentioned the
decidability result in the case of the subgraph ordering in
\cite{bdkss:undecidability-gts}, but did not treat it in detail and did
not consider it as an instance of a general framework.

In contrast to the minor ordering, the subgraph ordering is not a wqo on
the set of all graphs, but only on those graphs where the length of undirected
paths is bounded \cite{ding:subgraphs-wqo}. This results in a trade-off: while 
the stricter order allows us to consider all possible sets of graph 
transformation rules in order to obtain a decision procedure, we have to make 
sure to consider a system where only graphs satisfying this restriction are
reachable. Even if this condition is not satisfied, the procedure can
yield useful partial coverability results.  Also, it often terminates
without excluding graphs not satisfying the restriction (this is also
the case for our running example), producing exact results.  We make
these considerations precise by introducing \emph{$Q$-restricted} WSTSs, where 
the order need only be a wqo on $Q$. In general, one wants $Q$ to be as large 
as possible to obtain stronger statements.

It turns out that the results of \cite{JK08} can be transferred to
this new setting. Apart from the minor ordering and the subgraph
ordering, there are various other wqos that could be used
\cite{fhr:wqo-bounded-treewidth}, leading to different classes of
systems and different notions of coverability. In order to avoid
redoing the proofs for every case, we here introduce a general
framework which works for the case where the partial order can be
represented by graph morphisms, which is applicable to several
important cases.  Especially, we state conditions required to perform
the backwards search. We show that the case of the minor ordering can
be seen as a special instance of this general framework and show that
the subgraph and the induced subgraph orderings are also compatible.
Finally we present an implementation and give runtime results.
\opt{long}{The proofs can be found in the Appendix~\ref{sec:proofs}.}
\opt{short}{For the proofs we refer the reader to the extended version
  of this paper \cite{wsts-gts-framework:arxiv}.}

\section{Preliminaries}
\label{sec:preliminaries}

\subsection{Well-structured Transition Systems}
\label{sec:wsts}

We define an extension to the notion of WSTS as introduced in
\cite{acjt:general-decidability,fs:well-structured-everywhere}, a
general framework for decidability results for infinite-state systems,
based on well-quasi-orders.

\begin{definition}[Well-quasi-order and upward closure]\label{def:wqo}
  A quasi-order %\footnote{Note that a quasi-order is the same as a
  % preorder, i.e., a reflexive and transitive relation.}
  $\leq$ (over a set $X$) is a \emph{well-quasi-order (wqo)} if for
  any infinite sequence $x_0, x_1, x_2, \ldots$ of elements of $X$,
  there exist indices $i<j$ with $x_i \leq x_j$.

  An \emph{upward-closed set} is any set $I \subseteq X$ such that $x
  \leq y$ and $x \in I$ implies $y \in I$. For a subset $Y \subseteq
  X$, we define its upward closure $\upclosed{Y} = \{x \in X \mid \exists y
  \in Y \colon y \leq x\}$.  Then, a \emph{basis} of an upward-closed
  set $I$ is a set $I_B$ such that $I = \upclosed{I_B}$.
  A \emph{downward-closed set}, downward closure and a basis of a 
  downward-closed set can be defined analogously.
\end{definition}

%Wqos satisfy the following properties.
The definition of wqos gives rise to properties which are important for 
the correctness and termination of the backwards search algorithm presented 
later.

\begin{lemma}\label{lem:wqo-ucs}
  Let $\leq$ be a wqo, then the following two statements hold:
  \begin{enumerate}
    \item Any upward-closed set $I$ has a finite basis.
    \item For any infinite, increasing sequence of upward-closed sets $I_0 
    \subseteq I_1 \subseteq I_2 \subseteq \ldots$ there exists an index $k \in 
    \nat$ such that $I_i = I_{i+1}$ for all $i \geq k$.
  \end{enumerate}
\end{lemma}

A $Q$-restricted WSTS is a transition system, equipped with a quasi-order, such 
that the quasi-order is a (weak) simulation relation on all states and a wqo on 
a restricted set of states $Q$.

\begin{definition}[$Q$-restricted well-structured transition system]
\label{def:wsts}
Let $S$ be a set of states and let $Q$ be a downward closed subset of
$S$, where membership is decidable.  A \emph{$Q$-restricted well-structured 
transition system ($Q$-restricted WSTS)} is a transition system $\mathcal{T} = 
(S, \Rightarrow, \leq)$, where the following conditions hold:

\noindent
\parbox{0.7\textwidth}{%
  \begin{description}
    \item[Ordering:] $\leq$ is a quasi-order on $S$ and a wqo on $Q$.
    \item[Compatibility:] For all $s_{1} \leq t_{1}$ and a transition $s_{1} 
    \Rightarrow s_{2}$, there exists a sequence $t_{1} \Rightarrow^* t_{2}$ of 
    transitions such that $s_{2} \leq t_{2}$.
  \end{description}
}%
\parbox{0.3\textwidth}{%
  \begin{center}%
    \scalebox{1.0}{\input{diagrams/wsts-general.tex}}
  \end{center}%
}%
\end{definition}

The presented $Q$-restricted WSTS are a generalization of WSTS and are 
identical to the classical definition, when $Q = S$. We will show how 
well-known results for WSTS can be transfered to $Q$-restricted WSTS.
For $Q$-restricted WSTS there are two coverability problems of
interest. The \emph{(general) coverability problem} is to decide,
given two states $s,t \in S$, whether there is a sequence of
transitions $s \Rightarrow s_1 \Rightarrow \ldots \Rightarrow s_n$
such that $t \leq s_n$. The \emph{restricted coverability problem} is
to decide whether there is such a sequence for two $s,t \in Q$ with
$s_i \in Q$ for $1 \leq i \leq n$.  Both problems are undecidable in
the general case (as a result of \cite{bdkss:undecidability-gts} and
Proposition~\ref{prop:subgraph-is-undecidable}) but we will show that
the well-known backward search for classical WSTS can be put to good
use.

Given a set $I \subseteq S$ of states we denote by $\pred{I}$ the set
of direct predecessors of $I$, i.e., $\pred{I} = \{s \in S \mid
\exists s' \in I \colon s \Rightarrow s'\}$. Additionally, we use
$\pred[Q]{I}$ to denote the restriction $\pred[Q]{I} = \pred{I} \cap
Q$. Furthermore, we define $\predAll{I}$ as the set of all predecessors (in 
$S$) which can reach some state of $I$ with an arbitrary number of transitions. 
To obtain decidability results, the sets of predecessors must be computable,
i.e.~a so-called effective pred-basis must exist.

\begin{definition}[Effective pred-basis]
  \label{def:effective-pred-basis}
  A $Q$-restricted WSTS has an \emph{effective pred-basis} if there exists an
  algorithm accepting any state $s \in S$ and returning
  $\predBasis{s}$, a finite basis of $\upclosed{\pred{\upclosed{\{s\}}}}$.
  It has an \emph{effective $Q$-pred-basis} if there exists an algorithm 
  accepting any state $q \in Q$ and returning $\predBasis[Q]{q}$, a finite 
  basis of $\upclosed{\pred[Q]{\upclosed{\{q\}}}}$.
\end{definition}

Whenever there exists an effective pred-basis, there also exists an
effective $Q$-pred-basis, since we can use the downward closure of
$Q$ to prove $\predBasis[Q]{q} = \predBasis{q} \cap Q$. 
%This does not hold in the other direction.

Let $(S, \Rightarrow, \leq)$ be a $Q$-restricted WSTS with an effective 
pred-basis and let $I \subseteq S$ be an upward-closed set of states with 
finite basis $I_B$. To solve the general coverability problem we compute the 
sequence $I_0, I_1, I_2 \ldots$ where $I_0 = I_B$ and $I_{n+1} = I_n \cup 
\predBasis{I_n}$. If the sequence $\upclosed{I_0} \subseteq \upclosed{I_1} 
\subseteq \upclosed{I_2} \subseteq \ldots$ becomes stationary, i.e.~there is an 
$m$ with $\upclosed{I_m} = \upclosed{I_{m+1}}$, then $\upclosed{I_m} = 
\upclosed{\predAll{I}}$ and a state of $I$ is coverable from a state $s$ if and 
only if there exists an $s' \in I_m$ with $s' \leq s$. If $\leq$ is a wqo on 
$S$, by Lemma~\ref{lem:wqo-ucs} every upward-closed set is finitely 
representable and every sequence becomes stationary. However, in general the 
sequence might not become stationary if $Q \neq S$, in which case the problem 
becomes semi-decidable, since termination is no longer guaranteed (although 
correctness is).

The restricted coverability problem can be solved in a similar way, if
an effective $Q$-pred-basis exists. Let $I^Q \subseteq S$ be an upward
closed set of states with finite basis $I^Q_B \subseteq Q$. We
compute the sequence $I^Q_0, I^Q_1, I^Q_2, \ldots$ with $I^Q_0 =
I^Q_B$ and $I^Q_{n+1} = I^Q_n \cup \predBasis[Q]{I^Q_n}$. Contrary to
the general coverability problem, the sequence $\upclosed{I^Q_0} \cap
Q \subseteq \upclosed{I^Q_1} \cap Q \subseteq \upclosed{I^Q_2} \cap Q
\subseteq \ldots$ is guaranteed to become stationary according to
Lemma~\ref{lem:wqo-ucs}. Let again $m$ be the first index with
$\upclosed{I^Q_m} = \upclosed{I^Q_{m+1}}$, and set $\Rightarrow_Q\ = 
(\Rightarrow \cap\ Q \times Q)$.
We obtain the following result, of which the classical decidability
result of \cite{fs:well-structured-everywhere} is a special case.

\newcommand{\thmCoveringProblem}{
  Let $T = (S, \Rightarrow, \leq)$ be a $Q$-restricted WSTS with a
  decidable order $\leq$.
  \begin{enumerate}[(i)]
  \item \label{thm:covering-problem:case-all} If $T$ has an effective
    pred-basis and $S = Q$, the general and restricted coverability
    problems coincide and both are decidable.
  \item \label{thm:covering-problem:case-special} If $T$ has an
    effective $Q$-pred-basis, the restricted coverability problem is
    decidable if $Q$ is closed under reachability.
  \item \label{thm:covering-problem:case-general} If $T$ has an
    effective $Q$-pred-basis and $I^Q_m$ is the limit as described
    above, then: if $s \in \upclosed{I^Q_m}$, then $s$ covers a state
    of $I^Q$ in $\Rightarrow$ (general coverability). If $s \notin
    \upclosed{I^Q_m}$, then $s$ does not cover a state of $I^Q$ in
    $\Rightarrow_Q$ (no restricted coverability).
  \item \label{thm:covering-problem:case-hope} If $T$ has an effective
    pred-basis and the sequence $I_n$ becomes stationary for $n=m$,
    then: a state $s$ covers a state of $I$ if and only if $s \in
    \upclosed{I_m}$.
  \end{enumerate}
}
\newcommand{\thmCoveringProblemTitle}{Coverability problems}
\begin{theorem}[\thmCoveringProblemTitle]\label{thm:covering-problem}
  \thmCoveringProblem
\end{theorem}

Thus, if $T$ is a $Q$-restricted WSTS and the ``error states'' can be
represented as an upward-closed set $I$, then the reachability of an
error state of $I$ can be determined as described above, depending on
which of the cases of Theorem~\ref{thm:covering-problem} applies.
Note that it is not always necessary to compute the limits $I_m$ or
$I^Q_m$, since $\upclosed{I_i} \subseteq \upclosed{I_m}$ (and
$\upclosed{I_i^Q} \subseteq \upclosed{I_m^Q}$) for any $i \in \nat_0$.
Hence, if $s \in \upclosed{I_i}$ (or $s \in \upclosed{I_i^Q}$) for
some $i$, then we already know that $s$ covers a state of $I$ (or of
$I^Q$) in $\Rightarrow$.

\subsection{Graph Transformation Systems}\label{sec:gts}

In the following we define the basics of hypergraphs and GTSs as a
special form of transition systems where the states are hypergraphs
and the rewriting rules are hypergraph morphisms. We prefer
hypergraphs over directed or undirected graphs since they are more
convenient for system modelling.

\begin{definition}[Hypergraph]\label{def:hypergraph}
  Let $\Lambda$ be a finite sets of edge labels and $\arity \colon
  \Lambda \to \nat$ a function that assigns an arity to each label.  A
  \emph{($\Lambda$-)hypergraph} is a tuple $(V_{G}, E_{G}, c_{G},
  l_G^E)$ where $V_{G}$ is a finite set of nodes, $E_{G}$ is a finite
  set of edges, $c_{G}\colon E_{G} \rightarrow V_{G}^*$ is an (ordered)
  connection function and $l_G^E\colon E_{G} \rightarrow \Lambda$ is
  an edge labelling function. We require that $|c_G(e)| =
  \arity(l_G^E(e))$ for each edge $e \in E_G$.

  An edge $e$ is called \emph{incident} to a node $v$ (and vice versa) if $v$ 
  occurs in $c_G(e)$.
\end{definition}

From now on we will often call hypergraphs simply graphs.  An (elementary)
\emph{undirected path} of length $n$ in a hypergraph is an alternating
sequence $v_0, e_1, v_1, \dots, v_{n-1}, e_n, v_n$ of nodes and edges
such that for every index $1\le i \le n$ both nodes $v_{i-1}$ and
$v_i$ are incident to $e_i$ and the undirected path contains all nodes
and edges at most once.  Note that there is no established notion of
directed paths for hypergraphs, but our definition gives rise to
undirected paths in the setting of directed graphs (which are a
special form of hypergraphs).

% (Note that for directed graphs, our notion of paths corresponds to
% undirected paths, where the edge direction is ignored.)

\begin{definition}[Partial hypergraph morphism]\label{def:morphism}
Let $G$, $G'$ be ($\Lambda$-)hyper\-graphs. A \emph{partial hypergraph 
morphism} (or simply \emph{morphism}) $\phi \colon G\pto G'$ consists of a 
pair of partial functions $(\phi_{V}: V_{G} \pto V_{G'}, \phi_{E}: E_{G} \pto 
E_{G'})$ such that for every $e \in E_{G}$ it holds that $l_{G}(e) = 
l_{G'}(\phi_E(e))$ and $\phi_{V}(c_{G}(e)) = c_{G'}(\phi_{E}(e))$ whenever 
$\phi_E(e)$ is defined. Furthermore if a morphism is defined on an edge, it 
must be defined on all nodes incident to it.
\emph{Total morphisms} are denoted by an arrow of the form $\to$.
\end{definition}

For simplicity we will drop the subscripts and write $\phi$ instead of
$\phi_V$ and~$\phi_E$.  We call two graphs $G_1$, $G_2$ isomorphic if
there exists a total bijective morphism $\phi : G_1 \to G_2$.

Graph rewriting relies on the notion of \emph{pushouts}.  It is known that
pushouts of partial graph morphisms always exist and are unique up to
isomorphism.  Intuitively, for morphisms $\phi : G_0 \pto G_1$, $\psi
: G_0 \pto G_2$, the pushout is obtained by gluing the two graphs
$G_1,G_2$ over the common interface $G_0$ and by deleting all elements
which are undefined under $\phi$ or $\psi$ (for a formal definition see 
\opt{long}{Appendix~\ref{sec:pushouts}}\opt{short}{\cite{r:gra-handbook}}).

We will take pushouts mainly in the situation described in
Definition~\ref{def:rewriting} below, where $r$ (the rule) is partial
and connects the left-hand side $L$ and the right-hand side $R$. It is
applied to a graph $G$ via a total match $m$. In order to ensure that
the resulting morphism $m'$ (the co-match of the right-hand side in
the resulting graph) is also total, we have to require a match $m$ to be
\emph{conflict-free} wrt.\ $r$, i.e., if there are two elements $x,y$
of $L$ with $m(x)=m(y)$ either $r(x),r(y)$ are both defined or both
undefined.
Here we consider a graph rewriting approach called the
\emph{single-pushout approach (SPO)} \cite{ehklrwc:algebraic-approaches-II},
since it relies on one pushout square, and restrict to conflict-free matches.

\begin{definition}[Graph rewriting]\label{def:rewriting}\label{def:gts}
  A \emph{rewriting rule} is a partial morphism $r\colon L\pto R$,
  where $L$ is called left-hand and $R$ right-hand side.  A
  \emph{match} (of $r$) is a total morphism $m\colon L\to G$,
  conflict-free wrt.\ $r$.  Given a rule and a match, a
  \emph{rewriting step} or rule application is given by a pushout
  diagram as shown below, resulting in the graph $H$.

\noindent
\parbox{0.8\textwidth}{%
  A \emph{graph transformation system (GTS)} is a finite set of rules
  $\mathcal{R}$.  Given a fixed set of graphs $\mathcal{G}$, a
  \emph{graph transition system} on $\mathcal{G}$ generated by a graph
  transformation system $\mathcal{R}$ is represented by a tuple
  $(\mathcal{G},\Rightarrow)$ where $\mathcal{G}$ is the set of states
  and $G \Rightarrow G'$ if and only if $G,G' \in \mathcal{G}$ and $G$
  can be rewritten to $G'$ using a rule of $\mathcal{R}$.}
\parbox{0.2\textwidth}
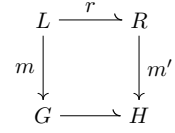
{%
\begin{center}
  \scalebox{\catDiagramScale}{\input{diagrams/rewriting.tex}}
\end{center}}
\end{definition}

Later we will have to apply rules backwards, which means that it is
necessary to compute so-called pushout complements, i.e., given $r$
and $m'$ above, we want to obtain $G$ (such that $m$ is total and
conflict-free). How this computation can be performed in general is
described in \cite{HJKS:pocs2010}.  Note that pushout complements are
not unique and possibly do not exist for arbitrary morphisms. For two
partial morphisms the number of pushout complements may be infinite.

% Mostly we will use GTSs and graph transition systems synonymously,
% e.g.~we call a GTS well-structured if the arising graph transition
% system is well-structured.

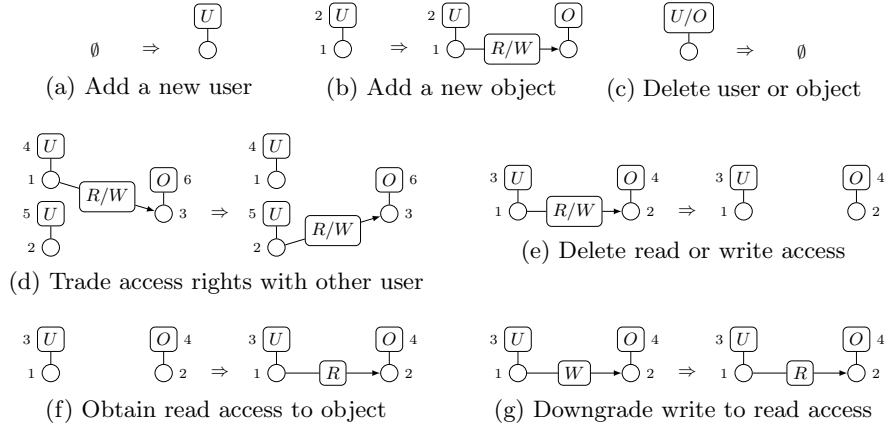
\begin{figure}[ht]
\centering
	\subfloat[Add a new user]{%
		\label{fig:main-example-add-user}
		\begin{minipage}[c]{0.27\textwidth}%
		\centering
		\scalebox{\ExampleRuleScale}{\input{diagrams/main-example-add-user.tex}}
		\end{minipage}}%
	\subfloat[Add a new object]{%
		\label{fig:main-example-add-object}
		\begin{minipage}[c]{0.35\textwidth}%
		\centering
		\scalebox{\ExampleRuleScale}{\input{diagrams/main-example-add-object.tex}}
		\end{minipage}}%
	\subfloat[Delete user or object]{%
		\label{fig:main-example-delete}
		\begin{minipage}[c]{0.27\textwidth}%
		\centering
		\scalebox{\ExampleRuleScale}{\input{diagrams/main-example-delete.tex}}
		\end{minipage}}\\
	\subfloat[Trade access rights with other user]{%
		\label{fig:main-example-trade-rights}
		\begin{minipage}[c]{0.5\textwidth}%
		\centering
		\scalebox{\ExampleRuleScale}{\input{diagrams/main-example-trade-rights.tex}}
		\end{minipage}}%
  \subfloat[Delete read or write access]{%
		\label{fig:main-example-remove-rights}
		\begin{minipage}[c]{0.5\textwidth}%
		\centering
		\scalebox{\ExampleRuleScale}{\input{diagrams/main-example-remove-rights.tex}}
		\end{minipage}}\\
	\subfloat[Obtain read access to object]{%
		\label{fig:main-example-get-read-right}
		\begin{minipage}[c]{0.5\textwidth}%
		\centering
		\scalebox{\ExampleRuleScale}{\input{diagrams/main-example-get-read-right.tex}}
		\end{minipage}}
	\subfloat[Downgrade write to read access]{%
		\label{fig:main-example-degrade-rights}
		\begin{minipage}[c]{0.5\textwidth}%
		\centering
		\scalebox{\ExampleRuleScale}{\input{diagrams/main-example-degrade-rights.tex}}
		\end{minipage}}%
\caption{A GTS modelling a multi-user system}
\label{fig:main-example}
\end{figure}

\begin{example}\label{example:rewriting}
  To illustrate graph rewriting we model a multi-user system as a GTS
  (see Figure~\ref{fig:main-example}) inspired by
  \cite{km:gb-models-access-control}. A graph contains user nodes,
  indicated by unary $U$-edges, and object nodes, indicated by unary
  $O$-edges. Users can have read~($R$) or write~($W$) access rights
  regarding objects indicated by a (directed) edge. Note that binary
  edges are depicted by arrows, the numbers describe the rule
  morphisms and labels of the form~$R/W$ represent two rules, one with
  $R$-edges and one with $W$-edges.

  The users and objects can be manipulated by rules for adding new
  users (Fig.~\ref{fig:main-example-add-user}), adding new objects
  with read or write access associated with a user
  (Fig.~\ref{fig:main-example-add-object}) and deleting users or
  objects (Fig.~\ref{fig:main-example-delete}). Both read and write
  access can be traded between users
  (Fig.~\ref{fig:main-example-trade-rights}) or dropped
  (Fig.~\ref{fig:main-example-remove-rights}). Additionally users can
  downgrade their write access to a read access
  (Fig.~\ref{fig:main-example-degrade-rights}) and obtain read access
  of arbitrary objects (Fig.~\ref{fig:main-example-get-read-right}).

  \begin{figure}[t]
    \begin{minipage}{0.25\textwidth}
      \centering
      \scalebox{\graphDiagramScale}{\input{diagrams/main-example-error-graph.tex}}
      \caption{An undesired state in the multi-user system}
      \label{fig:main-example-error-graph}
    \end{minipage}
    ~\hfill
    \begin{minipage}{0.75\textwidth}
      \centering
      \scalebox{\graphDiagramScale}{\input{diagrams/main-example-rewriting.tex}}
      \caption{Example of two rule applications}
      \label{fig:main-example-rewriting}
    \end{minipage}
  \end{figure}
  
  In a multi-user system there can be arbitrary many users with read
  access to an object, but at most one user may have write access.
  This means especially that any configuration of the system
  containing the graph depicted in
  Figure~\ref{fig:main-example-error-graph} is erroneous.

  An application of the Rules~\ref{fig:main-example-trade-rights}
  and~\ref{fig:main-example-get-read-right} is shown in
  Figure~\ref{fig:main-example-rewriting}. In general, nodes and edges
  on which the rule morphism~$r$ is undefined are deleted and nodes
  and edges of the right-hand side are added if they have no preimage
  under~$r$. In the case of non-injective rule morphisms, nodes or
  edges with the same image are merged.  Finally, node deletion
  results in the deletion of all incident edges (which would otherwise
  be left dangling).  For instance, if
  Rule~\ref{fig:main-example-delete} is applied, all read/write access
  edges attached to the single deleted node will be deleted as well.
\end{example}

\section{GTS as WSTS: A General Framework}
\label{sec:gts-as-wsts}

In this section we state some sufficient conditions such that the
coverability problems for $\mathcal{Q}$-restricted well-structured GTS
can be solved in the sense of Theorem~\ref{thm:covering-problem} (in
the following we use $\mathcal{Q}$ to emphasize that $\mathcal{Q}$ is
a set of graphs). We will also give an appropriate backward algorithm.
The basic idea is to represent the wqo by a given class of morphisms.

\begin{definition}[Representable by morphisms] 
  \label{def:condition-representation}
  Let $\genOrder$ be a quasi-order that satisfies $G_1\genOrder G_2$,
  $G_2\genOrder G_1$ for two graphs $G_1,G_2$ if and only if $G_1,G_2$
  are isomorphic, i.e., $\genOrder$ is anti-symmetric up to
  isomorphism. 

  We call $\genOrder$ \emph{representable by morphisms} if there is a
  class of (partial) morphisms $\mathcal{M}_{\genOrder}$ such that for
  two graphs $G,G'$ it holds that $G' \genOrder G$ if and only if
  there is a morphism $(\mu \colon G \genArrow G') \in
  \mathcal{M}_{\genOrder}$. Furthermore, for $(\mu_1 : G_1 \genArrow
  G_2), (\mu_2 : G_2 \genArrow G_3) \in \mathcal{M}_{\genOrder}$ it
  holds that $\mu_2 \circ \mu_1 \in \mathcal{M}_{\genOrder}$, i.e.,
  $\mathcal{M}_{\genOrder}$ is closed under composition. We call such
  morphisms $\mu$ \emph{order morphisms}.
\end{definition}

The intuition behind an order morphism is the following: whenever
there is an order morphism from $G$ to $G'$, we usually assume that
$G'$ is the smaller graph that can be obtained from $G$ by some form
of node deletion, edge deletion or edge contraction. For any graphs
$G$ (which represent all larger graphs) we can now compose rules and
order morphisms to simulate a co-match of a rule to some graph larger
than $G$. However, for this construction to yield correct results, the
order morphisms have to satisfy the following two properties.

\noindent
\parbox{0.8\textwidth}{%
  \begin{definition}[Pushout preservation]
    \label{def:condition-preservation} We say that a set of order
    morphisms $\mathcal{M}_{\genOrder}$ is \emph{preserved by total
      pushouts} if the following holds: if $(\mu \colon G_0 \genArrow
    G_1) \in \mathcal{M}_{\genOrder}$ is an order morphism and $g
    \colon G_0 \to G_2$ is total, then the morphism $\mu'$ in the
    pushout diagram on the right is an order morphism of
    $\mathcal{M}_{\genOrder}$.
  \end{definition}}
\parbox{0.2\textwidth}
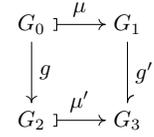
{%
\begin{center}
  \scalebox{\catDiagramScale}{\input{diagrams/preserving-order-morphisms.tex}}
\end{center}}

%The next property is needed to ensure that whenever a graph $H$,
%resulting from a rewriting step, is represented by a graph $S$
%($S\sqsubseteq H$), we can make a backwards rewriting step from $S$,
%obtaining a graph $G'$ that represents $G$.
The next property is needed to ensure that every graph $G$, which is rewritten 
to a graph $H$ larger than $S$, is represented by a graph $G'$ obtained by a 
backward rewriting step from $S$, i.e.~the backward step need not be applied to 
$H$.

\begin{definition}[Pushout closure] \label{def:condition-closure} Let
  $m \colon L \to G$ be total and conflict-free wrt.~$r\colon L\pto
  R$. A set of order morphisms is called \emph{pushout closed} if the
  following holds: if the diagram below on the left is a pushout and
  $\mu \colon H \genArrow S$ an order morphism, then there exist
  graphs $R'$ and $G'$ and order morphisms $\mu_R \colon R \genArrow
  R'$, $\mu_G \colon G \genArrow G'$, such that:
  \begin{enumerate}
  \item the diagram below on the right commutes and the outer square is a 
  pushout.
\item the morphisms $\mu_G \circ m \colon L \to G'$ and $n \colon R'
  \to S$ are total and $\mu_G \circ m$ is conflict-free wrt.~$r$.
\end{enumerate}
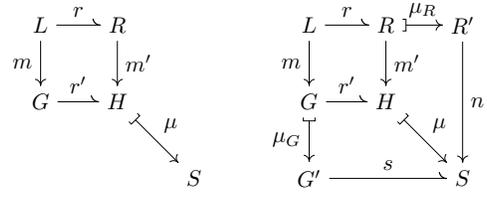
\begin{center}
  \scalebox{\catDiagramScale}{\input{diagrams/order-construction-closure.tex}}
\end{center}
\end{definition}

We now present a generic backward algorithm for (partially) solving
both coverability problems. The procedure has two variants, which both
require a GTS, an order and a set of final graphs to generate a set
of minimal representatives of graphs covering a final graph. The
first variant computes the sequence $I_n^Q$ and restricts the set of
graphs to ensure termination. It can be used for
cases~(\ref{thm:covering-problem:case-all}),~(\ref{thm:covering-problem:case-special})
and~(\ref{thm:covering-problem:case-general}) of
Theorem~\ref{thm:covering-problem}, while the second variant computes $I_n$ 
(without restriction) and can be used for
cases~(\ref{thm:covering-problem:case-all}) and
(\ref{thm:covering-problem:case-hope}).

\begin{procedure}[Computation of the ($\mathcal{Q}$-)pred-basis]
\label{procedure:main}\mbox{}

  \pseudoParagraph{Input} A set $\mathcal{R}$ of graph transformation
  rules, a quasi-order $\genOrder$ on all graphs which is a wqo on a 
  downward-closed set $\mathcal{Q}$ and a finite set of final graphs 
  $\mathcal{F}$, satisfying:
  \begin{itemize}
  \item The transition system generated by the rule set $\mathcal{R}$
    is a $\mathcal{Q}$-restricted WSTS with respect to the order
    $\genOrder$.
  \item The order $\genOrder$ is representable by a class of morphisms
    $\mathcal{M}_{\genOrder}$
    (Definition~\ref{def:condition-representation}) and this class
    satisfies Definitions~\ref{def:condition-preservation} and
    \ref{def:condition-closure}.
  \item \emph{Variant 1.} The set of minimal pushout complements
    restricted to $\mathcal{Q}$ with respect to $\genOrder$ is
    computable, for all pairs of rules and co-matches (it is
    automatically finite).
  
    \emph{Variant 2.} The set of minimal pushout complements with
    respect to $\genOrder$ is finite and computable, for all pairs of
    rules and co-matches.
  \end{itemize}

  \pseudoParagraph{Preparation} Generate a new rule set $\mathcal{R}'$
  from $\mathcal{R}$ in the following way: for every rule $(r : L \pto
  R) \in \mathcal{R}$ and every order morphism $\mu \colon R \genArrow
  \overline{R}$ add the rule $\mu \circ r$ to $\mathcal{R}'$. (Note
  that it is sufficient to take a representative $\overline{R}$ for
  each of the finitely many isomorphism classes, resulting in a finite
  set $\mathcal{R}'$.) Start with the working set $\mathcal{W} =
  \mathcal{F}$ and apply the first backward step.

  \pseudoParagraph{Backward Step} Perform backward steps until the
  sequence of working sets $\mathcal{W}$ becomes stationary. The
  following substeps are performed in one backward step for each rule
  $(r : L \pto R) \in \mathcal{R}'$:

  \begin{enumerate}
  \item \label{backstep:calculate-matchings} For a graph $G \in
    \mathcal{W}$ compute all total morphisms $m' : R \to G$
    (co-matches of $R$ in $G$).
  \item \label{backstep:calculate-pocs} \emph{Variant 1.} For each
    such morphism $m'$ calculate the set $\mathcal{G}_{poc}$ of
    minimal pushout complement objects of $m'$ with rule $r$, which
    are also elements of~$\mathcal{Q}$.

  \emph{Variant 2.} Same as Variant~1, but calculate \emph{all}
  minimal pushout complements, without the restriction to $\mathcal{Q}$.
\item \label{backstep:next-step} Add all remaining graphs in
  $\mathcal{G}_{poc}$ to $\mathcal{W}$ and minimize $\mathcal{W}$ by
  removing all graphs $G'$ for which there is a graph $G'' \in
  \mathcal{W}$ with $G' \neq G''$ and $G''\sqsubseteq G'$.
  \end{enumerate}

  \pseudoParagraph{Result} The resulting set $\mathcal{W}$ contains
  minimal representatives of graphs from which a final state is
  coverable (cf.~Theorem~\ref{thm:covering-problem}).

\ignore{
  \begin{enumerate}
    \item[] \emph{Variant 1.} If there is a graph $G \in \mathcal{W}$ with $G 
    \sqsubseteq G_0$, then $G_0$ covers a final graph within $\Rightarrow$. If 
    there is no such graph $G$, then $G_0$ does not cover a final graph within 
    $\Rightarrow_{\mathcal{Q}}$.
    \item[] \emph{Variant 2.} A graph $G_0$ covers a final graph if and only if 
    there is a graph $G \in \mathcal{W}$ with $G \sqsubseteq G_0$. (Note that 
    this variant might not terminate.)
  \end{enumerate}}
\end{procedure}

The reason for composing rule morphisms with order morphisms when
doing the backwards step is the following: the graph $G$, for which we
perform the step, might not contain a right-hand side $R$ in its
entirety. However, $G$ can represent graphs that do contain $R$ and hence
we have to compute the effect of applying the rule backwards to all
graphs represented by $G$. Instead of enumerating all these graphs
(which are infinitely many), we simulate this effect by looking for
matches of right-hand sides modulo order morphisms.
We show that the procedure is correct by proving the following lemma.

\newcommand{\propProcCorrectness}{Let \predBasis[1]{} and \predBasis[2]{} be a 
single backward step of Procedure~\ref{procedure:main} for Variant~1 and~2 
respectively. For each graph $S$, \predBasis[1]{S} is an effective 
$\mathcal{Q}$-pred-basis and \predBasis[2]{S} is an effective pred-basis.}
\begin{proposition}\label{prop:proc-correctness}
\propProcCorrectness
\end{proposition}

%%%%% proof in Appendix %%%%%

\section{Well-quasi Orders for Graph Transformation Systems}
\label{sec:wqo-for-gts}

\subsection{Minor Ordering}\label{sec:minor-ordering}

We first instantiate the general framework with the minor ordering,
which was already considered in \cite{JK08}.
The minor ordering is a well-known order on graphs, which is defined
as follows: a graph $G$ is a minor of $G'$ whenever $G$ can be
obtained from $G'$ by a series of node deletions, edge deletions and
edge contractions, i.e.~deleting an edge and merging its incident nodes 
according to an arbitrary partition. Robertson and Seymour showed in a seminal 
result that the minor ordering is a wqo on the set of all
graphs \cite{rs:graph-minors-xx}, even for hypergraphs
\cite{RS:graphMinors:XXIII}, thus case~(\ref{thm:covering-problem:case-all}) of 
Theorem~\ref{thm:covering-problem} applies. In 
\cite{JK08,jk:minor-wqo-corrected} we showed that the conditions for WSTS are 
satisfied for a restricted set of GTS by introducing minor morphisms and 
proving a result analogous to Proposition~\ref{prop:proc-correctness}, but only 
for this specific case. The resulting algorithm is a special case of both 
variants of Procedure~\ref{procedure:main}.

\begin{proposition}[\cite{JK08}]
  The coverability problem is decidable for every GTS if the minor ordering is 
  used and the rule set contains edge contraction rules for
  each edge label.
\end{proposition}

\subsection{Subgraph Ordering}
\label{sec:subgraph-ordering}

In this paper we will show that the subgraph ordering and the induced
subgraph ordering satisfy the conditions of
Procedure~\ref{procedure:main} for a restricted set
of graphs and are therefore also compatible with our framework. For
the subgraph ordering we already stated a related result (but for
injective instead of conflict-free matches) in
\cite{bdkss:undecidability-gts}, but did not yet instantiate a general
framework.

\begin{definition}[Subgraph]
  Let $G_1$, $G_2$ be graphs. $G_1$ is a subgraph of $G_2$ (written
  $G_1 \subOrder G_2$) if $G_1$ can be obtained from $G_2$ by a
  sequence of deletions of edges and isolated nodes.  We call a
  partial morphism $\mu \colon G \subArrow S$ a \emph{subgraph
    morphism} if and only if it is injective on all elements on which
  it is defined and surjective.
\end{definition}

It can be shown that the subgraph ordering is representable by
subgraph morphisms, which satisfy the necessary properties.  Using a
result from Ding \cite{ding:subgraphs-wqo} we can show that the set
$\mathcal{G}_k$ of hypergraphs where the length of every undirected
path is bounded by $k$, is well-quasi-ordered by the subgraph
relation. A similar result was shown by Meyer for depth-bounded
systems in \cite{m:structural-stationarity-pi}. Note that we bound
undirected path lengths instead of directed path lengths. For the class
of graphs with bounded directed paths there exists a sequence of
graphs violating the wqo property (a sequence of circles of increasing
length, where the edge directions alternate along the circle).

Since every GTS satisfies the compatibility condition of 
Definition~\ref{def:wsts} naturally, we obtain the following result.

\newcommand{\propSubgraphWSTS}{Let $k$ be a natural number. Every graph 
transformation system forms a $\mathcal{G}_k$-restricted WSTS with the 
subgraph ordering.}
\newcommand{\propSubgraphWSTSTitle}{WSTS wrt.~the subgraph ordering}
\begin{proposition}[\propSubgraphWSTSTitle]\label{prop:subgraph-wsts}
\propSubgraphWSTS
\end{proposition}

%%%%% proof is in the Appendix

The set of minimal pushout complements (not just restricted to $\mathcal{G}_k$) 
is always finite and can be computed as in the minor case.

\newcommand{\propSubgraphDecidable}{Every $\mathcal{G}_k$-restricted
  well-structured GTS with the subgraph order has an effective
  pred-basis and the (decidability) results of
  Theorem~\ref{thm:covering-problem} apply.}
\begin{proposition}\label{prop:subgraph-is-decidable}
\propSubgraphDecidable
\end{proposition}

%%%%% proof is in the Appendix

By a simple reduction from the reachability problem for two counter machines, 
we can show that the restricted coverability problem is undecidable in the 
general case. Although we cannot directly simulate the zero test, i.e.~negative 
application conditions are not possible, we can make sure that the rules 
simulating the zero test are applied correctly if and only if the bound $k$ was 
not exceeded.

\newcommand{\propSubgraphUndecidable}{Let $k > 2$ be a natural number. The 
restricted coverability problem for $\mathcal{G}_k$-restricted well-structured 
GTS with the subgraph ordering is undecidable.}
\begin{proposition}\label{prop:subgraph-is-undecidable}
\propSubgraphUndecidable
\end{proposition}

\begin{example}
  Now assume that an error graph is given and that a graph exhibits an error if 
  and only if it contains the error graph as a subgraph. Then we can use
  Proposition~\ref{prop:subgraph-is-decidable} to calculate all graphs
  which lead to some error configuration.

  For instance, let a multi-user system as described in
  Example~\ref{example:rewriting} be given. Normally we have to choose
  a bound on the undirected path length to guarantee termination, but
  in this example Variant~2 of Procedure~\ref{procedure:main}
  terminates and we can solve coverability on the unrestricted
  transition system (see
  Theorem~\ref{thm:covering-problem}(\ref{thm:covering-problem:case-hope})).
  The graph in Figure~\ref{fig:main-example-error-graph} represents
  the error in the system and by applying
  Procedure~\ref{procedure:main} we obtain a set of four graphs (one
  of which is the error graph itself), fully characterizing all
  predecessor graphs. We can observe that the error can only be
  reached from graphs already containing two $W$-edges going to a
  single object node. Hence, the error is not produced by the given
  rule set if we start with the empty graph and thus the system is
  correct.

  Interestingly the backward search finds the leftmost graph below due
  to the depicted sequence of rule applications, which leads directly to the 
  error graph. Thus, the error can occur even if a single user has two write 
  access rights to an object, because of access right trading.
  \begin{center}
    \scalebox{\graphDiagramScale}{\input{diagrams/subgraph-example.tex}}
  \end{center}
  The other two graphs computed are shown below and represent
  states with "broken" structure (a node cannot be a user \emph{and}
  an object). The left graph for instance can be rewritten to a graph
  larger than the left graph above, by a non-injective match of the
  rule in Figure~\ref{fig:main-example-trade-rights} mapping both
  nodes~2 and~3 to the right node.
  
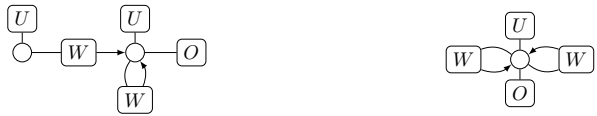
\begin{center}
   	\begin{minipage}[c]{0.45\textwidth}%
   	\centering
   	\scalebox{\graphDiagramScale}{\input{diagrams/example2-graph3.tex}}
   	\end{minipage}%
   	\begin{minipage}[c]{0.45\textwidth}%
   	\centering
   	\scalebox{\graphDiagramScale}{\input{diagrams/example2-graph4.tex}}
   	\end{minipage}%
  \end{center}
%  \begin{figure}[ht]
%  \centering
%  	\begin{minipage}[c]{0.45\textwidth}%
%  	\centering
%  	\img{\graphDiagramScale}{diagrams/example2-graph3}
%  	\end{minipage}%
%  	\begin{minipage}[c]{0.45\textwidth}%
%  	\centering
%  	\img{\graphDiagramScale}{diagrams/example2-graph4}
%  	\end{minipage}%
%  \caption{Two of four graphs calculated by the Procedure~\ref{procedure:main} 
%  for Example\ref{example:rewriting}}
%  \label{fig:subgraph-example-other-graphs}
%  \end{figure}
\end{example}

\subsection{Induced Subgraph Ordering}
\label{sec:induced-subgraph-ordering}

As for the subgraph ordering in Section~\ref{sec:subgraph-ordering}
our backward algorithm can also be applied to the induced subgraph
ordering, where a graph $G$ is considered as an induced subgraph of
$G'$ if every edge in $G'$ connecting only nodes also present in $G$,
is contained in $G$ as well. Unfortunately, this ordering is not a wqo
even when bounding the longest undirected path in a graph, such that
we also have to bound the multiplicity of edges between two nodes.
Note that this restriction is implicitly done in
\cite{ding:subgraphs-wqo} since Ding uses simple graphs.

Furthermore, since we do not know whether the induced subgraph
ordering can be extended to a wqo on (a class of)
hypergraphs, we here use only \emph{directed graphs}, where each edge
is connected to a sequence of exactly two nodes. For many applications
directed graphs are sufficient for modelling, also for our examples,
since unary hyperedges can simply be represented by loops.

At first, this order seems unnecessary, since it is stricter than the
subgraph ordering and is a wqo on a more restricted set of graphs. On
the other hand, it allows us to specify error graphs more precisely,
since a graph $G \in \mathcal{F}$ does not represent graphs with
additional edges between nodes of $G$. Furthermore one could equip the
rules with a limited form of negative application conditions, still
retaining the compatibility condition of Definition~\ref{def:wsts}.

\begin{definition}[Induced subgraph]
  Let $G_1$, $G_2$ be graphs. $G_1$ is an induced subgraph of $G_2$
  (written $G_1 \indsubOrder G_2$) if $G_1$ can be obtained from $G_2$
  by deleting a subset of the nodes and all incident edges.  We call a
  partial morphism $\mu \colon G \indsubArrow S$ an \emph{induced
    subgraph morphism} if and only if it is injective for all elements
  on which is defined, surjective, and if it is undefined on an edge
  $e$, it is undefined on at least one node incident to $e$.
\end{definition}

\newcommand{\propIndsubgraphWSTS}{ Let $n,k$ be natural numbers and
  let $\mathcal{G}_{n,k}$ be a set of directed, edge-labelled graphs,
  where the longest undirected path is bounded by $n$ and every two nodes are
  connected by at most $k$ parallel edges with the same label (bounded
  edge multiplicity).  Every GTS forms a
  $\mathcal{G}_{n,k}$-restricted WSTS with the induced subgraph
  ordering.}  \newcommand{\propIndsubgraphWSTSTitle}{WSTS wrt.~the
  induced subgraph ordering}
\begin{proposition}[\propIndsubgraphWSTSTitle]
\label{prop:induced-subgraphs-wsts}
\propIndsubgraphWSTS
\end{proposition}

%%%%% proof is in the Appendix

\newcommand{\propIndsubgraphDecidable}{Every $\mathcal{G}_{n,k}$-restricted 
well-structured GTS with the induced subgraph order has an effective 
$\mathcal{G}_{n,k}$-pred-basis and the (decidability) results of 
Theorem~\ref{thm:covering-problem} apply.}
\begin{proposition}\label{prop:induced-subgraph-is-decidable}
\propIndsubgraphDecidable
\end{proposition}

%%%%% proof is in the Appendix

The computation of minimal pushout complements in this case is
considerably more complex, since extra edges have to be added\opt{long}{ (see 
the proof of Proposition~\ref{prop:induced-subgraph-is-decidable} in
Appendix~\ref{sec:proofs})}, but we also obtain additional
expressiveness. In general GTS with negative application conditions do not 
satisfy the compatibility condition with respect to the subgraph relation, but 
we show in the following example, that it may still be satisfied with respect 
to the induced subgraph relation.

\begin{example}
  Let the following simple rule be given, where the negative
  application condition is indicated by the dashed edge, i.e.~the rule
  is applicable if and only if there is a matching only for the solid
  part of the left-hand side and this matching cannot be extended to
  match also the dashed part.

  \begin{center}
    \scalebox{\graphDiagramScale}{\input{diagrams/induced-subgraph-example.tex}}
  \end{center}
  
  Applied to a graph containing only $A$-edges, this rule calculates
  the transitive closure and will terminate at some point. This GTS
  satisfies the compatibility condition wrt.~the induced subgraph
  ordering, since for instance a directed path of length two (the
  left-hand side) does not represent graphs where there is an edge
  from the first to the last node of the graph.  Therefore we can use
  the induced subgraph ordering and our procedure to show that a graph
  containing two parallel $A$-edges can only be reached from graphs
  already containing two parallel $A$-edges.
\end{example}

The principle described in the example can be extended to all negative
application conditions which forbid the existence of edges but not of
nodes. This is the case, because if there is no edge between two nodes
of a graph, there is also no edge between these two nodes in any
larger graph. Hence if there is no mapping from the negative
application condition into the smaller graph, there can also be none
into the larger graph. Graphs violating the negative application
condition are simply not represented by the smaller graph. Hence, all
graph transformation rules with such negative application conditions
satisfy the compatibility condition wrt.\ the induced subgraph
ordering. The backward step has to be modified in this case by
dropping all obtained graphs which do not satisfy one of the negative
application conditions.

% As for the subgraph ordering, we can use the induced subgraph ordering
% to prove coverability properties of graph transformation systems. The
% expressible properties differ from the properties expressible using
% the subgraph ordering.

\subsection{Implementation}

We implemented Procedure~\ref{procedure:main} with support for the
minor ordering as well as the subgraph ordering in the tool
\textsc{Uncover}. The tool is written in C++ and designed in a modular
way for easy extension with further orders. The sole optimization
currently implemented is the omission of all rules that are also
order morphisms. It can be shown that the backward application of such
rules produces only graphs which are already represented.  
%A support for parallel computation is planned.

Table~\ref{fig:runtime-results} shows the runtime results of
different case studies, namely a leader election protocol and a
termination detection protocol (in an incorrect as well as a correct
version), using the minor ordering, and the access rights management
protocol described in Figure~\ref{fig:main-example} as well as a
public-private server protocol, using the subgraph order. It shows for
each case the restricted graph set $\mathcal{Q}$, the variant of the
procedure used (for the minor ordering they coincide), the runtime and
the number of minimal graphs representing all predecessors of error
graphs.

\begin{table}[h]
\caption{Runtime result for different case studies}
  \label{fig:runtime-results}
  \centering
  \begin{tabular}{l|c|c|c|r|r}
    \hline
    case study & wqo & graph set $\mathcal{Q}$ & variant & time & \#(error 
    graphs)
    \\\hline\hline
    Leader election & minor & all graphs & 1 / 2 & 3s & 38 \\\hline
    Termination detection (faulty) & minor & all graphs & 1 / 2 & 7s & 69 \\\hline
    Termination detection (correct) & minor & all graphs & 1 / 2 & 2s & 101 
    \\\hline
    Rights management & subgraph & all graphs & 2 & 1s & 4 \\\hline
    Public-private server ($l = 5$) & subgraph & path $\leq 5$ & 1 & 1s 
    & 14 \\\hline
    Public-private server ($l = 6$) & subgraph & path $\leq 6$ & 1 & 
    16s & 16 \\\hline
  \end{tabular}
\end{table}

\section{Conclusion}
\label{sec:conclusion}

We have presented a general framework for viewing GTSs as restricted
WSTSs. We showed that the work in \cite{JK08} for the minor ordering
can be seen as an instance of this framework and we presented two
additional instantiations, based on the subgraph ordering and the
induced subgraph ordering. Furthermore we presented the management of
read and write access rights as an example and discussed our
implementation with very encouraging runtime results. 

% We believe that the problem of checking whether $\mathcal{G}_k$ is
% closed under reachability for a given GTS is decidable, although we
% did not prove it.

Currently we are working on an extension of the presented framework
with rules, which can uniformly change the entire neighbourhood of
nodes. In this case the computed set of predecessor graphs will be an
over-approximation.  More extensions are possible (possibly
introducing over-approximations) and we especially plan to further
investigate the integration of rules with negative application
conditions as for the induced subgraph ordering.  In
\cite{ks:wsts-nac} we introduced an extension with negative
application conditions for the minor ordering, but still, the
interplay of the well-quasi-order and conditions has to be better
understood.  Naturally, we plan to look for additional orders, for
instance the induced minor and topological minor orderings
\cite{fhr:wqo-bounded-treewidth} in order to see whether they can be
integrated into this framework and to study application scenarios.

\smallskip

\noindent \textit{Related work.}  Related to our work is
\cite{BKWZ:2013}, where the authors use the subgraph ordering and a
forward search to prove fair termination for depth-bounded systems.
In \cite{abchr:monotonic-abstraction-heaps} another wqo
for well-structuring graph rewriting is considered, however only for
graphs where every node has out-degree~$1$. It would be interesting to
see whether this wqo can be integrated into our general
framework. The work in \cite{dsz:verification-ad-hoc-networks} uses
the induced subgraph ordering to verify broadcast protocols.
There the rules are different from our setting: a left-hand side
consists of a node and its entire neighbourhood of arbitrary size.
Finally \cite{swj:gg-verification-adhocrouting} uses a backwards
search on graph patterns in order to verify an ad-hoc routing
protocol, but not in the setting of WSTSs.

\smallskip

\noindent\textbf{Acknowledgements:} We would like to thank Roland Meyer, for
giving us the idea to consider the subgraph ordering on graphs, and
Giorgio Delzanno for several interesting discussions on wqos and
WSTSs.

\bibliography{backward-subgraph}
\bibliographystyle{plain}

\opt{long}{
\newpage
\appendix

\section{Pushouts}
\label{sec:pushouts}

We give the definition and construction of pushouts, the graph gluing
construction which is used in this paper.

\noindent
\parbox{0.7\textwidth}{
\begin{definition}\label{def:pushout}
  Let $\phi\colon G_0\pto G_1$ and $\psi\colon G_0\pto G_2$ be two
  partial graph morphisms. The \emph{pushout} of $\phi$ and $\psi$
  consists of a graph $G_3$ and two morphisms $\psi'\colon G_1\pto
  G_3$, $\phi'\colon G_2\pto G_3$ such that $\psi'\circ\phi =
  \phi'\circ\psi$ and for every other pair of morphisms $\psi''\colon
  G_1\pto G'_3$, $\phi''\colon G_2\pto G'_3$ such that
  $\psi''\circ\phi = \phi''\circ\psi$ there exists a unique morphism
  $\eta\colon G_3\pto G'_3$ with $\eta\circ\psi' = \psi''$ and
  $\eta\circ\phi' = \phi''$.
\end{definition}}
\parbox{0.3\textwidth}
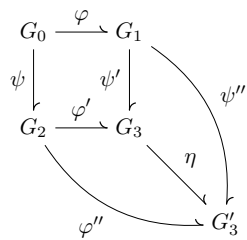
{
\begin{center}
  \scalebox{\catDiagramScale}{\input{diagrams/pushout.tex}}
\end{center}}

\begin{proposition}[Construction of pushouts]
  %\label{prop:construction-pushouts}
  Let $\phi\colon G_0\pto G_1$, $\psi\colon G_0\pto G_2$ be partial
  hypergraph morphisms. Furthermore let $\equiv_V$ be the smallest
  equivalence on $V_{G_1}\cup V_{G_2}$ and $\equiv_E$ the smallest
  equivalence on $E_{G_1}\cup E_{G_2}$ such that $\phi(x)\equiv
  \psi(x)$ for every element $x$ of $G_0$.
  
  An equivalence class of nodes is called \emph{valid} if it does not
  contain the image of a node $x$ of $G_0$ for which $\phi(x)$ or
  $\psi(x)$ are undefined. Similarly a class of edges is \emph{valid}
  if the analogous condition holds and furthermore all nodes incident
  to these edges are contained in valid equivalence classes.

  Then the pushout graph $G_3$ of $\phi$ and $\psi$ consists of all
  valid equivalence classes $[x]_\equiv$ of nodes and edges, where
  $l_{G_3}([e]_\equiv) = l_{G_i}(e)$ and $c_{G_3}([e]_\equiv) =
  [v_1]_\equiv \dots [v_k]_\equiv$ if $e\in E_{G_i}$ and $c_{G_i}(e) =
  v_1\dots v_k$. Furthermore the morphisms $\psi',\phi'$ map nodes and
  edges to their respective equivalence classes.
\end{proposition}

\begin{definition}[Pushout complement]
  Let $\phi : G_0 \pto G_1$ and $\psi' : G_1 \pto G_3$ be morphisms.
  We call the graph $G_2$ together with the morphisms $\psi : G_0 \pto
  G_2$ and $\phi' : G_2 \pto G_3$ a pushout complement, if $G_3$,
  $\phi'$, $\psi'$ is the pushout of $\phi$ and $\psi$.
\end{definition}

\section{Proofs}
\label{sec:proofs}

\subsection{Well-structured Transition Systems}

\begin{theorem_app}[\thmCoveringProblemTitle]{thm:covering-problem}
  \thmCoveringProblem
\end{theorem_app}

\begin{proof}
  (\ref{thm:covering-problem:case-all}) is just a reformulation of the
  decidability results for WSTS. Similar for
  (\ref{thm:covering-problem:case-special}), if $Q$ is closed under
  reachability, a $Q$-restricted WSTS can be seen as a WSTS with state
  space $Q$.

%  We now consider (\ref{thm:covering-problem:case-general}) where $Q$
%  is not required to be closed under reachability. Assume that
%  $q\in\upclosed{I^Q_m}$, where $I^Q_m$ has been obtained by
%  fixed-point iteration, starting with the upward-closed set $I^Q$ and
%  computing the sequence $I^Q_n$. Now set $I = I^Q$ and compute the
%  sequence $I_n$. Clearly $I^Q_n \subseteq I_n$. Hence $q\in I_m$,
%  which contains only states that can cover states in $I^Q = I$.  For
%  the other statement assume that $q\not\in \upclosed{I^Q_m}$ and
%  assume that there exists a path $q\Rightarrow_Q q_1 \Rightarrow_Q
%  \dots \Rightarrow_Q q_k\in I$. Then we can show by induction that
%  $q_i\in I_{k-i}$ and hence $q\in I_k\subseteq I_m$, which leads to a
%  contradiction.

  We now consider (\ref{thm:covering-problem:case-general}) where $Q$
  is not required to be closed under reachability. Assume that
  $s \in \upclosed{I^Q_m}$, where $I^Q_m$ has been obtained by
  fixed-point iteration, starting with the upward-closed set $I^Q_B$ and
  computing the sequence $I^Q_n$. By induction we show the existence of a 
  sequence of transitions leading from $s$ to some state in $\upclosed{I^Q_B}$.
  Obviously there is an $q_m \in I^Q_m$ with $q_m \leq s$ and by definition 
  either $q_m \in I^Q_{m-1}$ or there are $q_{m-1} \in I^Q_{m-1}$ and 
  $q_{m-1}'$ with $q_m \Rightarrow q_{m-1}'$ and $q_{m-1} \leq q_{m-1}'$. In 
  the latter case, because of the compatibility condition of 
  Definition~\ref{def:wsts}, there is a $q_{m-1}''$ with $s \Rightarrow^* 
  q_{m-1}''$ and $q_{m-1} \leq q_{m-1}' \leq q_{m-1}''$, i.e.~$s$ can reach an 
  element of $\upclosed{I^Q_{m-1}}$. Since this argument holds for 
  $q_{m-1}''$ as well, the state $s$ can ultimately reach a state $q_0'' \in 
  \upclosed{I^Q_B}$. Note that it is possible that $s = q_0''$, but it is not 
  guaranteed that $q_n'' \in Q$ for every $n$.

  For the other statement assume that $s \notin \upclosed{I^Q_m}$ and
  assume that there exists a path $s = q_0 \Rightarrow_Q q_1 \Rightarrow_Q
  \dots \Rightarrow_Q q_k \in \upclosed{I^Q_B}$. Note that the second 
  assumption is trivially false, if $s \notin Q$. We can show by induction and 
  by definition of \predBasis[Q]{} that 
  $q_i \in \upclosed{I^Q_{k-i}}$ and hence $q_0 \in \upclosed{I^Q_k} \subseteq 
  \upclosed{I^Q_m}$, which leads to a contradiction.

  The proof of case~(\ref{thm:covering-problem:case-hope}) is
  straightforward by observing that the set $I_m$ is an exact
  representation of all predecessors of $I$. \qed
\end{proof}

\subsection{GTS as WSTS: A General Framework}

\begin{lemma}\label{lem:proc-correctness1}
The sets generated by \predBasis[1]{S} and \predBasis[2]{S} are both finite 
subsets of \pred{\upclosed{\{S\}}} and $\predBasis[1]{S} \subseteq \mathcal{Q}$.
\end{lemma}

\begin{proof}
  By the conditions of Procedure~\ref{procedure:main}, the sets of minimal 
  pushout complements -- in the case of \predBasis[1]{S} restricted to 
  $\mathcal{Q}$ -- are finite and computable. Since the set of rules is also 
  finite, \predBasis[1]{S} and \predBasis[2]{S} are finite as well.
  Every non-minimal pushout complement in $\mathcal{Q}$ is represented by a 
  minimal pushout complement in $\mathcal{Q}$, because of the downward closure 
  of $\mathcal{Q}$. Thus, $\predBasis[1]{S} \subseteq \mathcal{Q}$ holds.

  Let $G \in \predBasis[1]{S} \cup \predBasis[2]{S}$ be a graph generated by 
  one of the procedures. Then there is a rule $r \colon L \to R$, an order 
  morphism $\mu \colon R \genArrow R'$ and a conflict-free match $m \colon L 
  \to G$, such that the left diagram below is a pushout.

  \begin{center}
    \scalebox{\catDiagramScale}{\input{diagrams/proc-correctness1.tex}}
  \end{center}

  Let $m' \colon R \to S'$, $r' \colon G \pto S'$ be the pushout of $m$, $r$. 
  Because the outer diagram on the right commutes, there is a unique morphism 
  $\mu' \colon S' \genArrow S$. The left and the outer square are both 
  pushouts and therefore also the right square is a pushout. 
  Since $m$ is total and conflict-free, $m'$ is also total. By assumption 
  $\mathcal{M}_{\genOrder}$ is preserved by total pushouts, thus $\mu'$ is in 
  fact an order morphism. This means that $G$ can be rewritten to some graph 
  larger than $S$, hence $G \in \pred{\upclosed{S}}$. \qed
\end{proof}

\begin{lemma}\label{lem:proc-correctness2}
It holds that $\upclosed{\predBasis[1]{S}} \supseteq 
\upclosed{\pred[\mathcal{Q}]{\upclosed{\{S\}}}}$ and
$\upclosed{\predBasis[2]{S}} \supseteq \upclosed{\pred{\upclosed{\{S\}}}}$.
\end{lemma}

\begin{proof}
  Let $G$ be an element of $\upclosed{\pred{\upclosed{\{S\}}}}$. Then there is 
  a minimal representative $G_1 \in \pred{\upclosed{\{S\}}}$ with $G_1 
  \genOrder G$ and a rule $r \colon L \pto R$ rewriting $G_1$ with a 
  conflict-free match $m$ to some element $G_2$ of $\upclosed{\{S\}}$. 
  According to Definition~\ref{def:condition-closure} the left diagram below 
  can be extended to the right diagram below.

  \begin{center}
    \scalebox{\catDiagramScale}{\input{diagrams/proc-correctness2.tex}}
  \end{center}

  Since the outer square is a pushout, $G_3$ is a pushout complement object. 
  Thus, a graph $G_4 \genOrder G_3$ will be obtained by the 
  procedure \predBasis[2]{} in Step~\ref{backstep:calculate-pocs} using the 
  rule $\mu_R \circ r$. Summarized, this means that \predBasis[2]{} computes a 
  graph $G_4$ for every graph $G$ such that $G_4 \genOrder G_3 \genOrder G_1 
  \genOrder G$, i.e.~every $G$ is represented by an element of \predBasis[2]{S}.
  
  Now assume $G \in \upclosed{\pred[\mathcal{Q}]{\upclosed{\{S\}}}}$. By
  definition, the minimal representative $G_1$ is an element of
  $\mathcal{Q}$.  We obtain $G_3 \in \mathcal{Q}$, due to the downward
  closure of $\mathcal{Q}$. Thus, the procedure \predBasis[1]{} will
  compute a graph $G_4 \genOrder G_3$ (with $G_4 \in \mathcal{Q}$),
  i.e.~every $G$ is represented by an element of $\predBasis[1]{S}$.
  \qed
\end{proof}

\begin{proposition_app}{prop:proc-correctness}
\propProcCorrectness
\end{proposition_app}

\begin{proof}
The correctness of \predBasis[1]{} and \predBasis[2]{} is a direct consequence 
of Lemma~\ref{lem:proc-correctness1} 
and~\ref{lem:proc-correctness2}. Moreover, by the conditions of 
Procedure~\ref{procedure:main}, the set of minimal pushout complements 
(possibly restricted to $\mathcal{Q}$) is finite and computable. Thus, 
\predBasis[1]{} and \predBasis[2]{} are effective. \qed
\end{proof}

\subsection{Subgraph Ordering}

\begin{lemma}\label{lem:subgraph-condition1}
  The subgraph ordering is representable by subgraph morphisms.
\end{lemma}

\begin{proof}
Let $S \subseteq G$, then by definition $S$ can be obtained from $G$ by a 
sequence of node and edge deletions of length $n$. For each $G_i$ and each 
node or edge $x \in G_i$ we can give a subgraph morphism $\mu_i \colon G_i 
\subArrow G_i \setminus \{x\}$, where $\mu_i$ is undefined on $x$ and the 
identity on all other elements. Note that a node can only be deleted if it has 
no incident edges. Since injectivity and surjectivity are preserved by 
concatenation, the concatenation $\mu = \mu_1 \circ \ldots \circ \mu_n$ is 
again a subgraph morphism.

Let $\mu \colon G \subArrow S$ be a subgraph morphism. Since $\mu$ is
surjective and injective, the inverse of $\mu$ is a total, injective
morphism $\mu^{-1} \colon S \to G$. The image of $\mu^{-1}$ is
isomorphic to $S$ and a subgraph of $G$, therefore $S \subseteq G$
holds. \qed
\end{proof}

\begin{lemma}\label{lem:subgraph-condition2}
  Subgraph morphisms are preserved by total pushouts.
\end{lemma}

\begin{proof}
  Let $g \colon G_0 \to G_2$ be a total morphism and let $\mu \colon G_0
  \subArrow G_1$ be a subgraph morphism, such that $\mu'$, $g'$ is the
  pushout of $\mu$, $g$ as shown in the diagram below.
  \begin{center}
    \scalebox{\catDiagramScale}{\input{diagrams/preserving-subgraph-morphisms.tex}}
  \end{center}
  First we show that $\mu'$ is injective where it is defined. Assume
  there are two different elements $x_1, x_2 \in G_2$ such that
  $\mu'(x_1) = \mu'(x_2)$. For $G_3$ to be a pushout, both $x_1$ and
  $x_2$ have to have preimages $x_1', x_2' \in G_0$ with $g(x_1') =
  x_1$ and $g(x_2') = x_2$.  The diagram commutes, thus $\mu$ is
  defined for both elements and these elements are mapped injectively
  to $x_1'', x_2'' \in G_1$ respectively.  Hence, there is a commuting
  diagram with $g'(x_1'') = \mu'(x_1) \neq g'(x_2'') = \mu'(x_2)$,
  where there is no mediating morphism from $G_3$. Since this violates
  the pushout properties of the diagram, $\mu'$ has to be injective.

  It remains to be shown that $\mu'$ is surjective. Assume there is an
  $x_3 \in G_3$ without a preimage under $\mu'$. For the diagram to be a
  pushout there has to be an $x_1 \in G_1$ with $g'(x_1) = x_3$. Since
  $\mu$ is surjective, there is an $x_0 \in G_0$ with $\mu(x_0) = x_1$ and
  therefore the diagram does not compute, because $g'(\mu(x_0))$ is
  defined, but $\mu'(g(x_0))$ is not. Hence, $\mu'$ has to be surjective
  and is a subgraph morphism. \qed
\end{proof}

\begin{lemma} \label{lem:subgraph-condition3}
  Subgraph morphisms are pushout closed.
\end{lemma}

\begin{proof}
Let the morphisms be given as in the left diagram below. We will show the 
existence of the subgraph morphisms $\mu_R$, $\mu_G$ and the morphisms 
$n$, $s$, such that $\mu_G \circ m$ and $n$ are total and 
$\mu_G \circ m$ is conflict-free wrt.~$r$.
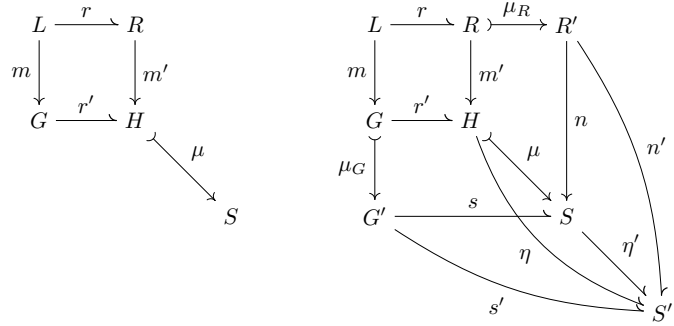
\begin{center}
  \scalebox{\catDiagramScale}{\input{diagrams/subgraph-construction-closure.tex}}
\end{center}
We define $R' = (V_{R'}, E_{R'}, c_{R'}, l_{R'})$ with $V_{R'} = \{v 
\in V_R \mid \mu(m'(v)) \text{ is defined}\}$, $E_{R'} = \{e 
\in E_R \mid \mu(m'(e)) \text{ is defined}\}$, $c_{R'}(e) = c_R(e)$ and 
$l_{R'}(e) = l_R(e)$ for all $e \in E_{R'}$. Note that all vertices of the 
sequence $c_R(e)$ are in $V_{R'}$ since $\mu \circ m'$ can only be 
defined for $e$ if it is defined for all attached vertices. On this basis we 
define $\mu_R(x) = x$ for all $x \in V_{R'} \cup E_{R'}$ (undefined otherwise) 
and $n = \mu \circ m'$. 
Obviously, $\mu_R$ is injective and surjective where it is defined, hence a 
subgraph morphism, and 
$n$ is total, since $\mu \circ m'$ is by definition defined on all 
elements of $R'$. Additionally by definition $\mu \circ m' = n \circ 
\mu_R$, since $\mu_R$ is undefined if and only if $\mu \circ m'$ is 
undefined.

The graph $G' = (V_{G'}, E_{G'}, c_{G'}, l_{G'})$ is defined in a similar way 
with $V_{G'} = \{v \in V_G \mid \exists v' \in V_L: m(v') = v \lor 
\mu(r'(v)) \text{ is defined}\}$, $E_{G'} = \{e \in E_G \mid \exists e' 
\in E_L: m(e') = e \lor \mu(r'(e)) \text{ is defined}\}$, $c_{G'}(e) 
= c_G(e)$ and $l_{G'}(e) = l_G(e)$ for all $e \in E_{G'}$. Note that $e \in 
E_{G'}$ again implies that all vertices of $c_G(e)$ are in $V_{G'}$. Also 
$\mu_G(x) = x$ for all $x \in V_{G'} \cup E_{G'}$ (undefined otherwise) and 
$s = \mu \circ r'$. Since $\mu_G$ is injective and
surjective where it is defined, it is 
a subgraph morphism. Additionally $\mu_G$ is defined on all elements of $G$ 
which have a preimage in $L$, hence $\mu_G \circ m$ is total and also 
conflict-free wrt.~$r$, since $\mu_G$ is injective and $m$ is 
conflict-free wrt.~$r$. By definition $\mu \circ r' = s \circ 
\mu_G$ since $\mu \circ r'$ is undefined on every element of $G$ on which 
$\mu_G$ is undefined.

Finally we show that the outer square is a pushout. We first observe that the 
outer diagram commutes, since $n \circ \mu_R \circ r = \mu \circ 
m' \circ r = \mu \circ r' \circ m = s \circ \mu_G 
\circ m$. Assume there is a graph $S'$ with morphisms $n' : R' \pto 
S'$ and $s' : G' \pto S'$ and $n' \circ \mu_R \circ r = s' 
\circ \mu_G \circ m$. The inner square is a pushout, hence there is a 
morphism $\eta : H \to S'$ such that $\eta \circ m' = n' \circ 
\mu_R$ and $\eta \circ r' = s' \circ \mu_G$. Since $\mu$ is 
injective and surjective, the inverse morphism $\mu^{-1} : S \to H$ is total 
and well-defined. Thus, there is a unique morphism $\eta' : S \pto 
S'$ with $\eta' = \eta \circ \mu^{-1}$. Due to the commutativity in the 
diagram we know that $s' \circ \mu_G = \eta \circ r' = \eta' \circ 
\mu \circ r' = \eta' \circ s \circ \mu_G$. Since $\mu_G$ is 
surjective, we obtain that $s' = \eta' \circ s$. Analogously we can 
show that $n' = \eta' \circ n$ commutes and the diagram is a
pushout. \qed
\end{proof}

\begin{lemma}
\label{lem:subgraph-compatibility}
Every GTS satisfies the compatibility condition of
Definition~\ref{def:wsts} with respect to the subgraph ordering.
\end{lemma}

\begin{proof}
  We have to show that whenever $G\Rightarrow H$ and $G\subOrder G'$,
  then there exists $H'$ with $G'\Rightarrow^* H'$ (here even
  $G'\Rightarrow H'$) and $G'\subOrder H'$.

  Let $m : L \pto R$ be a rule and $m : L \to G$ a matching that is
  conflict-free wrt.~$r$ such that $G$ is rewritten to $H$, i.e.~the
  upper inner square below is a pushout. Furthermore let $\mu : G'
  \subArrow G$ be a subgraph morphism, then the inverse morphism
  $\mu^{-1}$ is total and injective. Thus, the morphism $m_\mu : L \to
  G'$ with $m_\mu = \mu^{-1} \circ m$ in the diagram below is total
  and conflict-free wrt.~$r$ and $G'$ can be rewritten to $H'$.
  \begin{center}
    \scalebox{\catDiagramScale}{\input{diagrams/wsts-with-subgraphs.tex}}
  \end{center}
  Since the outer square is a 
  pushout, there is a unique morphism $\mu' : H \to H'$ such that the lower 
  inner square commutes. Furthermore since the upper inner square and the outer 
  square are pushouts, so is the lower inner square. This means that $\mu'$ is 
  total and injective since $\mu^{-1}$ is and pushouts preserve this 
  properties. Hence, $H \subOrder H'$. \qed
\end{proof}

\begin{proposition}\label{prop:subgraph-wqo}
  The subgraph ordering on hypergraphs is a wqo for the
  set of graphs where the longest undirected path is bounded by a constant.
\end{proposition}

\begin{proof}
  In \cite{ding:subgraphs-wqo} Ding showed that this proposition holds
  for undirected, simple graphs with node labels. We will now give an
  encoding $f$ of hypergraph to such graphs satisfying the following
  conditions:
  \begin{itemize}
  \item There is a function $g : \nat \to \nat$ such that, if the longest 
  undirected path in a hypergraph $G$ has length $k$, then the longest 
  undirected path in $f(G)$ has length $g(k)$.
  \item For every two hypergraphs $G_1$, $G_2$ if $f(G_1) \subOrder f(G_2)$ 
    then $G_1 \subOrder G_2$.
  \end{itemize}
  If these two properties hold, every infinite sequence $G_0, G_1, \ldots$ of
  hypergraphs with bounded undirected paths can be encoded into an
  infinite sequence of undirected graphs with bounded paths $f(G_0), f(G_1),
  \ldots$ of which we know that two elements $f(G_i) \subOrder f(G_j)$
  exist. Thus, also $G_i \subOrder G_j$ holds.

  Let $G = (V, E, c, l)$ be a $\Lambda$-hypergraph. We define its
  encoding as an undirected graph $f(G) = G' = (V', E', l')$ where
  $E'$ consists of two-element subsets of $V'$ and $l' : V' \to
  \Lambda'$ where the components are defined as follows:
  \begin{align*}
    V' &= V \cup E \cup \{(v,i,e) \mid v \in V, e \in E: c(e) = \alpha v \beta 
    \land |\alpha| = i\}\\
    E' &= \{\{x,y\} \mid x = (v,i,e) \in V' \land (y = v \lor y = e)\}\\
    \Lambda' &= \Lambda \cup \{N\} \cup \{n \in \nat_0 \mid \exists k \in 
    \Lambda : n < \arity(k)\}\\
    l'(x) &= \begin{cases}
      N & \text{if } x \in V \\
      l(x) & \text{if } x \in E \\
      i & \text{if } x = (v,i,e)
    \end{cases}
  \end{align*}
  Note that we assume that $N \notin \Lambda $ and $\Lambda \cap
  \nat_0 = \emptyset$.  An example of such an encoding can be seen in
  the diagram below, where the hypergraph on the left is encoded in
  the graph on the right-hand side.
  \begin{center}
    \scalebox{\graphDiagramScale}{\input{diagrams/subgraph-translation.tex}}
  \end{center}
  We now show that the encoding satisfies the two necessary
  properties.  First we observe, that every (undirected) graph
  generated by this encoding can be transformed back to a unique
  hypergraph, up to isomorphism.

  Now let $G$ be a hypergraph, where the longest undirected path is bounded by
  $k$, we show by contradiction that in $f(G)$ there can not be a path
  of length at least $4k+10$. Assume there is such a path in $f(G)$.
  Apart from the first or last node, all nodes labelled $N$ or $l \in
  \Lambda$ on this path are adjacent to nodes labelled with $n \in
  \nat_0$ and all nodes labelled with $n \in \nat_0$ are adjacent to
  (exactly) one node labelled with $N$ and one node labelled with $l
  \in \Lambda$. We now shorten the path in the least possible way to
  obtain a path of length at least $4k+4$ which starts and ends with
  nodes labelled with $N$. This path can be translated back to a
  sequence $v_0, e_1, v_1, \ldots, v_n, e_{n+1}, v_{n+1}$ since every
  node labelled with $N$ is a node of (the hypergraph) $G$ and every
  node labelled with $l \in \Lambda$ is an edge of $G$. This violates
  our assumption, that the longest undirected path of $G$ is bounded by $k$,
  thus, there is no path of length $4k+10$ or longer in $f(G)$.

  Let $G_1$, $G_2$ be hypergraphs such that $f(G_1) \subOrder f(G_2)$.
  Then there is a total, injective morphism $\mu : f(G_1) \to f(G_2)$.
  Since $f(G_i)$ contains (as nodes) all nodes and edges of $G_i$ (for
  $i \in \{1,2\}$), $\mu$ can be restricted to $V_{G_1} \cup E_{G_1}$
  and is then a total, injective morphism $\mu' : G_1 \to G_2$. The
  nodes of $f(G_i)$ labelled with natural numbers, ensures the
  morphism property on the hypergraphs. By inverting $\mu'$ we obtain
  an injective and surjective, but partial morphism from $G_2$ to
  $G_1$ (a subgraph morphism, see Lemma~\ref{lem:subgraph-condition1}), hence 
  $G_1 \subOrder G_2$. \qed 
\end{proof}

\begin{proposition_app}[\propSubgraphWSTSTitle]{prop:subgraph-wsts}
\propSubgraphWSTS
\end{proposition_app}

\begin{proof}
This is a direct consequence of Lemma~\ref{lem:subgraph-compatibility} and 
Proposition~\ref{prop:subgraph-wqo}. \qed
\end{proof}

\begin{proposition_app}{prop:subgraph-is-decidable}
\propSubgraphDecidable
\end{proposition_app}

\begin{proof}
  In Lemma~\ref{lem:subgraph-condition1},
  \ref{lem:subgraph-condition2} and \ref{lem:subgraph-condition3} we
  have shown that the subgraph ordering satisfies the conditions of
  Procedure~\ref{procedure:main}. Furthermore, the set of minimal
  pushout complements -- not just restricted to $\mathcal{G}_k$ -- can
  be computed in the same way as it is done in \cite{JK08} for the
  minor ordering, such that both variants of
  Procedure~\ref{procedure:main} are applicable. \qed
\end{proof}

\begin{proposition_app}{prop:subgraph-is-undecidable}
\propSubgraphUndecidable
\end{proposition_app}

\begin{proof}
We reduce the control state reachability problem of Minsky machines to the 
restricted coverability problem using the subgraph ordering on the set of 
graphs $\mathcal{G}_2$, where the length of the longest undirected path is less 
than or equal to two. Let 
$(Q,\Delta,(q_0,m,n))$ be the Minsky machine. We define a GTS using $\{q, q^B 
\mid q \in Q\} \cup \{c_1, c_2, X\}$ as the set of labels. The initial graph 
is shown in Figure~\ref{fig:subgraph-undecidable-initial} and illustrates how 
configurations of the Minsky machine are represented as graphs.

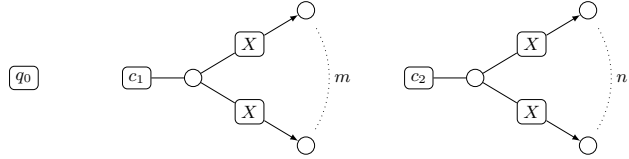
\begin{figure}[ht]
  \centering
  \scalebox{\graphDiagramScale}{\input{diagrams/subgraph-undecidable-initial.tex}}
  \caption{The initial configuration of the Minsky machine represented by a 
  graph}
  \label{fig:subgraph-undecidable-initial}
\end{figure}

For each transition rule of the Minsky machine, we add a graph transformation 
rule as shown in Figure~\ref{fig:subgraph-undecidable-rules}. A counter is 
represented as a star-like structure with the counters main node as centre, 
where the value of the counter is the number of attached $X$-edges. 
Incrementing and decrementing corresponds to creating and deleting $X$-edges. 
Regardless of the counters value, the longest undirected path of this structure 
has at most length two.

The zero-test adds two $X$-edges and blocks the state-edge, such that the 
rewritten graph has a undirected path of length three if and only if the 
counter was not zero (i.e.~had an $X$-edge attached). The auxiliary rules 
unblock the state to enable further computation.

\begin{figure}[ht]
  \centering
  \scalebox{\graphDiagramScale}{\input{diagrams/subgraph-undecidable-rules.tex}}
  \caption{Translation of Minsky rules to GTS rules}
  \label{fig:subgraph-undecidable-rules}
\end{figure}
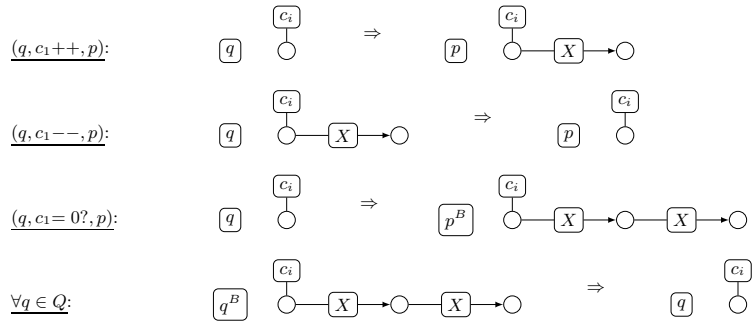

Obviously, if there is a sequence of transitions of the Minsky machine
which leads from a configuration $(q_0,m,n)$ to a state $q_f$, this
sequence can be copied in the GTS and every graph generated through
this sequence is in $\mathcal{G}_2$. On the other hand, if the graph
consisting of a single $q_f$-edge is $\mathcal{G}_2$-restricted coverable in 
the GTS, there is a sequence of rule applications corresponding to a sequence of
transitions of the Minsky machine. Since this rule applications
generate only graphs in $\mathcal{G}_2$, the zero-test-rule is only
applied if the counters value is in fact zero and the sequence of
transitions is valid. 

Instead of adding and removing a path of length two in the last two rules of 
Figure~\ref{fig:subgraph-undecidable-rules} one can add and remove a path of 
length $k$ to show the undecidability for $\mathcal{G}_k$-restricted 
well-structured GTS.\qed
\end{proof}

\subsection{Induced Subgraph Ordering}

\begin{lemma}\label{lem:indsubgraph-condition1}
The induced subgraph ordering is presentable by induced subgraph
morphisms.
\end{lemma}

\begin{proof}
Let $\mu_1 : G_1 \indsubArrow G_2$ and $\mu_2 : G_2 \indsubArrow G_3$ be two 
induced subgraph morphisms. Induced subgraph morphisms are closed under 
composition, since injectivity and surjectivity are preserved and if $\mu_2 
\circ \mu_1$ is undefined for some edge $e$, then $\mu_1$ is undefined on $e$ 
or $\mu_2$ is undefined on $\mu_1(e)$ implying that $\mu_2 \circ \mu_1$ is 
undefined for at least one node of $e$.

For some graph $G$ we can obtain any induced subgraph $G'$ by a sequence of 
node deletions including all attached edges. Each morphisms  
$\mu_i : G_i \indsubArrow G_{i+1}^x$ of this sequence, where $G_{i+1}^x$ is 
obtained by deleting the node $x$ and all its attached edges from $G_i$, is an 
induced subgraph morphisms and since they are closed under composition, the 
entire sequence is as well.

On the other hand every induced subgraph morphism $\mu : G \indsubArrow G'$ 
can be split into a sequence of node deletions (deleting all attached edges), 
since every deleted edge is attached to a deleted node, hence $G' \indsubOrder 
G$. \qed
\end{proof}

\begin{lemma}\label{lem:indsubgraph-condition2}
Induced subgraph morphisms are preserved by total pushouts.
\end{lemma}

\begin{proof}
  Since every induced subgraph morphism is also a subgraph morphism,
  $\mu'$ is injective and surjective by
  Lemma~\ref{lem:subgraph-condition2}.  Let $e \in E_{G_2}$ be an edge
  on which $\mu'$ is undefined. $e$ has a preimage $e' \in G_0$ since
  otherwise the pushout of $\mu$ and $g$ would contain $e$.  Since
  $\mu'(g(e'))$ is undefined, so is $g'(\mu(e'))$. In fact $\mu$ is
  undefined for $e'$ because otherwise $g$ and $\mu$ would be defined
  on $e$ and $e$ would be in the pushout. $\mu$ is an induced subgraph
  morphism, thus at least one of the nodes $v$ attached to $e'$ is
  undefined and also $\mu'$ has to be undefined on $g(v)$ for the
  diagram to commute. \qed
\end{proof}

\begin{lemma}\label{lem:indsubgraph-condition3}
Induced subgraph morphisms are pushout closed.
\end{lemma}

\begin{proof}
In Lemma~\ref{lem:subgraph-condition3} we have shown that there are subgraph 
morphisms $\mu_R$ and $\mu_G$ if $\mu$ is a subgraph morphism. We will show 
that these morphisms are induced subgraph morphisms if $\mu$ is an induced 
subgraph morphism.

Let $e \in E_R$ be an edge on which $\mu_R$ is undefined. By definition 
$\mu(m'(e))$ is undefined and since $m'$ is total, $\mu$ is 
undefined on $m'(e)$ (which is defined). Hence, at least one node $v$ 
attached to $m'(e)$ has no image under $\mu$ and all its preimages under 
$m'$ (which exist since $e$ has a preimage) are undefined under $\mu_R$. 
Thus, $e$ is attached to at least one node on which $\mu_R$ is undefined on.

Let $e' \in E_G$ be an edge on which $\mu_G$ is undefined. By definition 
$\mu(r'(e'))$ is undefined and $e'$ has no preimage under $m$. 
Because of the latter property, $e'$ is in the pushout $H$ and therefore 
defined under $r'$. Thus, $\mu$ is undefined on $r'(e)$ and on at 
least one attached node. All preimages under $r'$ of this node are 
undefined under $\mu_G$ since the diagram commutes. Hence, $\mu_G$ is an 
induced subgraph morphism. \qed
\end{proof}

\begin{lemma}\label{lem:induced-subgraphs-compatibility}
Every GTS satisfies the compatibility condition of 
Definition \ref{def:wsts} with respect to the induced subgraph ordering.
\end{lemma}

\begin{proof}
We modify the proof of Lemma~\ref{lem:subgraph-compatibility} by additionally 
showing that the reverse of $\mu'$ is an induced subgraph morphism.
Assume there is an edge $e \in E_{H'}$, where all attached nodes have a 
preimage under $\mu'$ but $e$ has none. Since $r'$, $\mu^{-1}$, $\mu'$, 
$r_\mu$ is a pushout, this can only be the case if all nodes attached to $e$ 
have a preimage in $G$ and $G'$ and $e$ has a preimage in $G'$. Because $\mu$ 
is an induced subgraph morphism, $e$ has a preimage in $G$. Due to 
commutativity $r'$ cannot be undefined on this preimage, thus, $e$ has to have 
a preimage in $H$, violating the assumption. \qed
\end{proof}

To prove Proposition~\ref{prop:induced-subgraph-wqo} we adapt Dings proof
using the notion of type of a graph.

\begin{definition}[Type of a Graph]
  A graph which consists of at most a single node with possibly
  attached edges has type one. A connected graph containing at least
  two nodes has at most type $n$, if there is a node $v$ so that the
  deletion of $v$ and all attached edges splits the graph into
  components which each have the type $n-1$. The type of a
  non-connected graph is the maximal type of its components.
\end{definition}

\begin{lemma}[\cite{ding:subgraphs-wqo}]\label{lem:pathlength-to-type}
  Every directed graph, where the longest undirected path has length $n$, has at
  most type $n+2$.
\end{lemma}

Note that contrary to Ding our type is bounded by $n+2$ instead of
$n$, because we measure path lengths via the number of edges instead
of nodes and Ding excludes paths of length $n$ to obtain graphs of
type at most $n$.

\begin{proposition}\label{prop:induced-subgraph-wqo}
  Let $n$, $k$ be natural numbers. The induced subgraph ordering is a
  wqo on the set of directed, edge-labelled graphs, where the longest undirected
  path is bounded by $n$ and every two nodes are connected by at most
  $k$ parallel edges with the same label (bounded edge multiplicity).
\end{proposition}

\begin{proof}
We prove this proposition by induction over the type of a graph, adapting 
Dings proof in \cite{ding:subgraphs-wqo} that undirected, node-labelled graphs 
of bounded type are well-quasi-ordered by the induced subgraph order. Because 
of Lemma~\ref{lem:pathlength-to-type} we know that the result for bounded 
types automatically transfers to bounded undirected paths.
To prove this proposition we use hypergraphs which are additionally node 
labelled, i.e.~there is a second alphabet $\Sigma$ of node labels and a 
(total) labelling function $\sigma : V_G \to \Sigma$. We obtain classical 
directed graphs if $|\Sigma| = 1$.

Let $G_1, G_2, \ldots$ be an infinite sequence of graphs of type $n$
and with edge multiplicity bounded by $k$. If $n=1$ then every $G_i$
consists of a single node with up to $k \cdot |\Lambda|$ attached 
loops. Since the sets of node and edge labels are finite, there are only 
finitely many possibilities to attach up to $k \cdot |\Lambda|$ edges to the 
node, thus $G_i \indsubOrder G_j$ for some $i < j$, i.e.~$\indsubOrder$ is a
wqo on the set of all such graphs.

Now let $n > 1$. Then there is a node $v_i \in G_i$ such that the
deletion of $v_i$ (and its attached edges) splits the graph into
components $G_{i,q}$ (for $1 \leq q \leq \ell_i$) of type at most
$n-1$. We define $\widetilde{G}_i$ to be the graph containing only
$v_i$ and its attached loops. Additionally we define
$\widehat{G}_{i,q}$ to be $G_{i,q}$ where the label $\sigma(y)$ of
every node $y$ is changed to $\sigma'(y) = (f_y,\sigma(y))$, where
$f_y: \Lambda \to \{0,1,\ldots,k\}^2$ is a function such that
$f_y(\lambda) = (a,b)$ where $a$ is the number of incoming and $b$ of
outgoing $\lambda$-labelled edges attached to both $y$ and $v_i$.
Since there are only finitely many possible functions $f_y$ (due to
the multiplicity constraint), the set of labels remains finite. We
extend $\indsubOrder$ to sequences such that $(\widetilde{G}_i,
\widehat{G}_{i,1}, \ldots, \widehat{G}_{\ell_i}) \indsubOrder^*
(\widetilde{G}_j, \widehat{G}_{j,1}, \ldots, \widehat{G}_{j,\ell_j})$
if and only if $\widetilde{G}_i \indsubOrder \widetilde{G}_j$ and
there are $p_1, \ldots, p_{\ell_i}$ with $1 \leq p_1 < \ldots <
p_{\ell_i} \leq \ell_j$ such that $\widehat{G}_{i,q} \indsubOrder
\widehat{G}_{j,p_q}$.  As shown for the case $n = 1$, $\indsubOrder$
is a wqo on all $\widetilde{G}_i$ and since the graphs
$\widehat{G}_{i,q}$, $\widehat{G}_{j,p_q}$ are of type $n-1$, they are
well-quasi-ordered by induction hypothesis. Hence, due to Higman
\cite{higman1952-divisibility} $\indsubOrder^*$ is also a wqo and
there are indices $i < j$ such that $(\widetilde{G}_i,
\widehat{G}_{i,1}, \ldots, \widehat{G}_{i,\ell_i}) \indsubOrder^*
(\widetilde{G}_j, \widehat{G}_{j,1}, \ldots, \widehat{G}_{j,\ell_j})$.
It remains to be shown that this implies $G_i \indsubOrder G_j$.  By
Lemma~\ref{lem:indsubgraph-condition1} there are induced subgraph
morphisms $\mu_0 : \widetilde{G}_j \indsubArrow \widetilde{G}_i$ and
$\mu_q : \widehat{G}_{j,p_q} \indsubArrow \widehat{G}_{i,q}$ for $1
\leq q \leq \ell_i$. We define the morphism $\mu : G_j \to G_i$ as
\[\mu(x) = \begin{cases}
v_i & \text{if } x = v_j\\
\mu_q(x) & \text{if } x \in \widehat{G}_{j,p_q} \text{ for some } q\\
\mu_0(x) & \text{if } x \in E_{\widetilde{G}_j} \text{ and } 
c_{\widetilde{G}_j}(x) = v_jv_j\\
\mu^v(x) & \text{if } x \in E_{G_j} \text{ and } c_{G_j}(x) = v_jv \lor 
c_{G_j}(x) = vv_j\\
& \text{for } v_j \neq v  \in V_{G_j} \text{ and } \mu(v) \text{ is defined}\\
\text{undefined} & \text{else}
\end{cases}\]
where $\mu^v$ is any total, bijective morphism from $G_j$ restricted to $v_j$, 
$v$ and the edges between them to $G_i$ restricted to $v_i$, $\mu_q(v)$ if $v 
\in \widehat{G}_{j,p_q}$ and any edges between them (both not including 
loops). Note that $\mu^v$ exists since $v$ and $\mu_q(v)$ are labelled with 
some $(f,\alpha)$, thus the number of edges between $v_j$ and $v$ is equal to 
the number of edges between $v_i$ and $\mu_q(v)$ for all labels and directions.

We now show that $\mu$ is a induced subgraph morphism. First note that
$\mu$ is a valid morphism since $\mu_q$, $\mu_0$ and $\mu^v$ are
morphisms and labels of edges in $G_i$, $G_j$ are the same as their
representative in $\widehat{G}_{j,p_q}$, $\widehat{G}_{i,q}$ and
representatives of nodes are labelled with $(f,\alpha)$ while the
origin is labelled $\alpha$ also implying equality on labels.  We then
observe that $\mu$ is injective and surjective, since $\mu_q$, $\mu_0$
and $\mu^v$ are all injective and surjective and $v_j$ is mapped to
$v_i$.  Assume there is an edge $e \in E_{G_j}$ for which $\mu$ is
undefined. If $e$ is contained in one of the components
$\widehat{G}_{j,p_q}$ or in $\widetilde{G}_j$, at least one attached
node is undefined, since $\mu_q$ and $\mu_0$ are an induced subgraph
morphisms. If $e$ connects $v_j$ and a node $v$ of a component
$\widehat{G}_{j,z}$, then either $z$ is not of the form $p_q$ and $\mu$ is
undefined on $\widehat{G}_{j,z}$ or $z = p_q$ and $\mu(v)$ is
undefined since otherwise $\mu^v$ has a mapping for $e$. Since $\mu$ is an
induced subgraph morphism, we obtain that $G_i \indsubOrder G_j$. \qed
\end{proof}

\begin{proposition_app}[\propIndsubgraphWSTSTitle]{prop:induced-subgraphs-wsts}
\propIndsubgraphWSTS
\end{proposition_app}

\begin{proof}
This is a direct consequence of Lemma~\ref{lem:induced-subgraphs-compatibility} 
and Proposition~\ref{prop:induced-subgraph-wqo}. \qed
\end{proof}

\begin{proposition_app}{prop:induced-subgraph-is-decidable}
\propIndsubgraphDecidable
\end{proposition_app}

\begin{proof}
  As shown in Lemma~\ref{lem:indsubgraph-condition1},
  \ref{lem:indsubgraph-condition2} and
  \ref{lem:indsubgraph-condition3} the induced subgraph ordering
  satisfies the conditions of
  Procedure~\ref{procedure:main}.
  The computation of minimal pushout complements is more involved than
  in the subgraph case. This is due to the fact that if a rule deletes
  a node, all attached edges are deleted, even if these edges have no
  preimage in $L$.  Adding an edge to a pushout complement and
  attaching it to a node which is deleted by the rule, results in
  another pushout complement. Contrary to the subgraph ordering these
  pushout complements are not already represented by the graph without
  the edge, if the induced subgraph ordering is used, but we can
  compute them as follows:
  \begin{enumerate}
    \item Let $(r : L \pto R) \in \mathcal{R}'$ be a rule and $m : R \to G$ a 
    match calculated in Step 1 of Procedure~\ref{procedure:main}. Calculate 
    the set of minimal pushout complements $\mathcal{G}_{poc}$ wrt.~the 
    subgraph ordering restricted to $\mathcal{G}_{n,k}$.
    \item \label{step:indsubgraphs-minimal-pocs}
    For all pushout complement objects $X \in \mathcal{G}_{poc}$ with 
    morphism $r' : X \pto G$, add all $X'$ to $\mathcal{G}_{poc}$, where $X'$ 
    can be obtained by adding an edge to $X$, which is attached to at least 
    one node on which $r'$ is undefined. Do not add $X'$ if it exceeds the 
    bounded multiplicity.
  \item Perform Step~\ref{step:indsubgraphs-minimal-pocs} until
    $\mathcal{G}_{poc}$ becomes stationary, which will be the case
    since the multiplicity is bounded. The set $\mathcal{G}_{poc}$ is
    then the set of minimal pushout complement objects wrt.~the
    induced subgraph ordering.
  \end{enumerate} 
  \qed
\end{proof}
}

\end{document}

%% file: tikz-defs.tex
\usepackage{etoolbox}
\usepackage{tikz}
\usetikzlibrary{positioning}
\usetikzlibrary{patterns}
\usetikzlibrary{arrows}
\usetikzlibrary{scopes}
\usetikzlibrary{backgrounds,fit}
\usetikzlibrary{calc}
\usetikzlibrary{decorations.markings}
\usetikzlibrary{shapes}

\newcommand{\tikzArrowToText}[1]{%
\mathop{\begin{tikzpicture}[baseline={([yshift=-0.5ex]current 
bounding box.south)}]%
\draw[#1] (0,0) -- (1.1em,0);%
\end{tikzpicture}}}

\tikzset{nodelabel/.style={font=\footnotesize}}
\tikzset{edgelabel/.style={font=\footnotesize}}
\tikzset{henodenumber/.style={font=\footnotesize}}
\tikzset{undiredge/.style={shorten >=0pt, shorten <=0pt}}
\tikzset{partialedge/.style={-left to, thick}}
\tikzset{dotline/.style={dotted, shorten <=5pt, shorten >=5pt}}

\tikzset{graphedge/.style={->, >=latex}}
\tikzset{bigraphedge/.style={<->, >=latex}}
\tikzset{invgraphedge/.style={<-, >=latex}}
\tikzset{henode/.style={draw, rectangle, rounded corners=2.5pt}}
\tikzset{stdnode/.style={draw, circle}}
\tikzset{mapedge/.style={->, >=stealth', dashed, color=orange}}

\tikzset{ruleappedge/.style={double, ->}}
\tikzset{mor-tot-inj/.style={catt-catt}}
\tikzset{mor-tot/.style={-catt}}
\tikzset{mor-parr/.style={-catpr}}
\tikzset{mor-parl/.style={-catpl}}
\tikzset{mor-minor/.style={|-catt}}
\tikzset{mor-subgraph/.style={)-catt}}
\tikzset{mor-indsubgraph/.style={open triangle 45 reversed-catt}}
\tikzset{mor-genorder/.style={]-catt}}
\tikzset{noedge/.style={decoration={markings,mark=at position 0.5 with 
{\arrow[scale=1.6,color=red]{|}}}, postaction={decorate}}}
\tikzset{closure/.style={postaction={decorate, decoration={raise=4pt, 
markings, mark=at position 1 with {\node[scale=0.65]{\textbf{*}};}}}}}

\tikzset{backstepedge/.style={-open triangle 45, shorten <=0pt, shorten >=0pt}}

\tikzset{graphbox/.style={draw, dashed, rounded corners=2mm, outer sep=5pt, on 
background layer}}
\tikzset{graphboxgrey/.style={draw, dashed, black!50, rounded corners=2mm, 
outer sep=5pt, on background layer}}

\newdimen\arrowsize
\pgfarrowsdeclare{recatt}{catt}
{
\arrowsize=1pt
\advance\arrowsize by .5\pgflinewidth
\pgfarrowsleftextend{-2.5\arrowsize-.5\pgflinewidth}
\pgfarrowsrightextend{.5\pgflinewidth}
}
{
\pgfsetlinewidth{0.4pt}
\arrowsize=0.8pt
\advance\arrowsize by .5\pgflinewidth
\pgfsetdash{}{0pt} % do not dash
\pgfsetroundjoin % fix join
\pgfsetroundcap % fix cap
\pgfpathmoveto{\pgfpointorigin}
\pgfpatharc{270}{190}{2.8\arrowsize}
\pgfusepathqstroke
\pgfpathmoveto{\pgfpointorigin}
\pgfpatharc{90}{170}{2.8\arrowsize}
\pgfusepathqstroke
}

\pgfarrowsdeclarereversed{catt}{recatt}{recatt}{catt}
\pgfarrowsdeclaredouble{recatcatt}{catcatt}{recatt}{catt}
\pgfarrowsdeclaredouble{catcatt}{recatcatt}{catt}{recatt}

\pgfarrowsdeclare{helperpr}{helperpr}
{
\arrowsize=1pt
\advance\arrowsize by .5\pgflinewidth
\pgfarrowsleftextend{-2.5\arrowsize-.5\pgflinewidth}
\pgfarrowsrightextend{.5\pgflinewidth}
}
{
\pgfsetlinewidth{0.55pt}
\arrowsize=0.8pt
\advance\arrowsize by .5\pgflinewidth
\pgfsetdash{}{0pt} % do not dash
\pgfsetroundjoin % fix join
\pgfsetroundcap % fix cap
\pgfpathmoveto{\pgfpointorigin}
\pgfpatharc{90}{170}{3\arrowsize}
\pgfusepathqstroke
}

\pgfarrowsdeclare{helperpl}{helperpl}
{
\arrowsize=1pt
\advance\arrowsize by .5\pgflinewidth
\pgfarrowsleftextend{-2.5\arrowsize-.5\pgflinewidth}
\pgfarrowsrightextend{.5\pgflinewidth}
}
{
\pgfsetlinewidth{0.55pt}
\arrowsize=0.8pt
\advance\arrowsize by .5\pgflinewidth
\pgfsetdash{}{0pt} % do not dash
\pgfsetroundjoin % fix join
\pgfsetroundcap % fix cap
\pgfpathmoveto{\pgfpointorigin}
\pgfpatharc{270}{190}{3\arrowsize}
\pgfusepathqstroke
}

\pgfarrowsdeclarereversed{rehelperpr}{rehelperpr}{helperpr}{helperpr}
\pgfarrowsdeclarereversed{rehelperpl}{rehelperpl}{helperpl}{helperpl}
\pgfarrowsdeclarealias{catpl}{catpl}{helperpl}{helperpl}
\pgfarrowsdeclarealias{catpr}{catpr}{helperpr}{helperpr}
\makeatletter
\@ifclassloaded{beamer}{%
\newcommand<>{\StaticGraphbox}[5][grbox]{%
\def\hAlignCJS{right}%
\def\vAlignCJS{top}%
\def\nameCJS{}%
\renewcommand{\do}[1]{%
  \ifthenelse{\equal{##1}{top}}{\def\vAlignCJS{top}}{%
  \ifthenelse{\equal{##1}{bottom}}{\def\vAlignCJS{bottom}}{%
  \ifthenelse{\equal{##1}{right}}{\def\hAlignCJS{right}}{%
  \ifthenelse{\equal{##1}{left}}{\def\hAlignCJS{left}}{\def\nameCJS{##1}}}}}}%
  \docsvlist{#5}%
  \begin{pgfonlayer}{background}
    \node[circle, above left= #4 and #3 of #2] (top_left) {};
    \node[circle, above right= #4 and #3 of #2] (top_right) {};
    \node[circle, below left= #4 and #3 of #2] (bottom_left) {};
    \node[circle, below right= #4 and #3 of #2] (bottom_right) {};
    \node[fit=(top_left) (bottom_right), outer sep=5pt] (#1) {} ;
    \draw#6[dashed, black!50,rounded corners=2mm] 
    ($(#1.north west) + (5pt,-5pt)$) -- ($(#1.north east) + (-5pt,-5pt)$) --
    ($(#1.south east) + (-5pt,5pt)$) -- ($(#1.south west) + (5pt,5pt)$) -- 
    cycle ;
    \node#6[left] (boxname) at (\vAlignCJS_\hAlignCJS.east) {\nameCJS};
  \end{pgfonlayer}
}}{%
\newcommand{\StaticGraphbox}[5][grbox]{%
\def\hAlignCJS{right}%
\def\vAlignCJS{top}%
\def\nameCJS{}%
\renewcommand{\do}[1]{%
  \ifthenelse{\equal{##1}{top}}{\def\vAlignCJS{top}}{%
  \ifthenelse{\equal{##1}{bottom}}{\def\vAlignCJS{bottom}}{%
  \ifthenelse{\equal{##1}{right}}{\def\hAlignCJS{right}}{%
  \ifthenelse{\equal{##1}{left}}{\def\hAlignCJS{left}}{\def\nameCJS{##1}}}}}}%
  \docsvlist{#5}%
  \begin{pgfonlayer}{background}
    \node[circle, above left= #4 and #3 of #2] (top_left) {};
    \node[circle, above right= #4 and #3 of #2] (top_right) {};
    \node[circle, below left= #4 and #3 of #2] (bottom_left) {};
    \node[circle, below right= #4 and #3 of #2] (bottom_right) {};
    \node[fit=(top_left) (bottom_right), outer sep=5pt] (#1) {} ;
    \draw[dashed, black!50,rounded corners=2mm] 
    ($(#1.north west) + (5pt,-5pt)$) -- ($(#1.north east) + (-5pt,-5pt)$) --
    ($(#1.south east) + (-5pt,5pt)$) -- ($(#1.south west) + (5pt,5pt)$) -- 
    cycle ;
    \node[left] (boxname) at (\vAlignCJS_\hAlignCJS.east) {\nameCJS};
  \end{pgfonlayer}
}}
\makeatother

%% file: general-defs.tex
\usepackage{hyperref}
\usepackage{ifthen}
\usepackage{rotating}

\newcommand{\nat}{\ensuremath{\mathbb{N}}}

\renewcommand{\phi}{\varphi}
\renewcommand{\epsilon}{\varepsilon}
\newcommand{\pto}{\rightharpoonup}

\newcommand{\upclosed}[1]{\mathord{\uparrow} #1}

\newcommand{\rotleq}{\begin{sideways}$\leq$\end{sideways}}

\newenvironment{theorem_app}[2][]{\noindent{\bf \hyperref[#2]{Theorem 
\ref*{#2}}\ifthenelse{\equal{#1}{}}{.}{ (#1).}}\it}{}
\newenvironment{proposition_app}[2][]{\noindent{\bf \hyperref[#2]{Proposition 
\ref*{#2}}\ifthenelse{\equal{#1}{}}{.}{ (#1).}}\it}{}

%% file: diagrams/wsts-general.tex
\begin{tikzpicture}[x=1.5cm,y=-1.0cm]

%%%%%%%%%%%%%%%%%%%% objects %%%%%%%%%%%%%%%%%%%%
\node (t1) at (0,0) {$t_1$};
\node (t2) at (1,0) {$t_2$};
\node (s1) at (0,1) {$s_1$};
\node (s2) at (1,1) {$s_2$};
\node (leq1) at ($(t1)!0.5!(s1)$) {\rotleq};
\node (leq2) at ($(t2)!0.5!(s2)$) {\rotleq};

%%%%%%%%%%%%%%%%%%%% arrows %%%%%%%%%%%%%%%%%%%%
\draw[ruleappedge,closure] (t1) -- (t2);
\draw[ruleappedge] (s1) -- (s2);

\end{tikzpicture}%

%% file: diagrams/rewriting.tex
\begin{tikzpicture}[x=1.5cm,y=-1.5cm]

\node (L) at (0,0) {$L$};
\node (R) at (1,0) {$R$};
\node (G) at (0,1) {$G$};
\node (H) at (1,1) {$H$};
\draw[mor-parl] (L) -- node [midway, above] {$r$} (R);
\draw[mor-tot] (L) -- node [midway, left] {$m$} (G);
\draw[mor-tot] (R) -- node [midway, right] {$m'$} (H);
\draw[mor-parl] (G) -- (H);

\end{tikzpicture}%

%% file: diagrams/main-example-add-user.tex
\begin{tikzpicture}[StdGraphGrid]

%%%%%%%%%%%%%%%%%%%% L %%%%%%%%%%%%%%%%%%%%
\node at (0,0) {$\emptyset$};

%%%%%%%%%%%%%%%%%%%% arrow %%%%%%%%%%%%%%%%%%%%
\node at (0.5,0) {$\Rightarrow$};

%%%%%%%%%%%%%%%%%%%% R %%%%%%%%%%%%%%%%%%%%
\begin{scope}[shift={(1,0)}]
  \node[stdnode] (rn1) at (0,0) {};
  \node[henode] (re1) at (0,-0.5) {$U$};
  \draw[undiredge] (rn1) -- (re1);
\end{scope}

\end{tikzpicture}%

%% file: diagrams/main-example-add-object.tex
\begin{tikzpicture}[StdGraphGrid]

%%%%%%%%%%%%%%%%%%%% L %%%%%%%%%%%%%%%%%%%%
\begin{scope}[shift={(0,0)}, every label/.style={font=\scriptsize}]
  \node[stdnode, label=left:{$1$}] (ln1) at (0,0) {};
  \node[henode, label=left:{$2$}] (le1) at (0,-0.5) {$U$};
  \draw[undiredge] (ln1) -- (le1);
\end{scope}

%%%%%%%%%%%%%%%%%%%% arrow %%%%%%%%%%%%%%%%%%%%
\node at (0.5,0) {$\Rightarrow$};

%%%%%%%%%%%%%%%%%%%% R %%%%%%%%%%%%%%%%%%%%
\begin{scope}[shift={(1,0)}, every label/.style={font=\scriptsize}]
  \node[stdnode, label=left:{$1$}] (rn1) at (0,0) {};
  \node[stdnode] (rn2) at (1,0) {};
  \node[henode, label=left:{$2$}] (re1) at (0,-0.5) {$U$};
  \node[henode] (re2) at (1,-0.5) {$O$};
  \node[henode] (re3) at ($(rn1)!0.5!(rn2)$) {$R/W$};
  \draw[undiredge] (rn1) -- (re1);
  \draw[undiredge] (rn2) -- (re2);
  \draw[graphedge] (rn1) -- (re3) -- (rn2);
\end{scope}

\end{tikzpicture}%

%% file: diagrams/main-example-delete.tex
\begin{tikzpicture}[StdGraphGrid]

%%%%%%%%%%%%%%%%%%%% L %%%%%%%%%%%%%%%%%%%%
\begin{scope}[shift={(0,0)}]
  \node[stdnode] (ln1) at (0,0) {};
  \node[henode] (le1) at (0,-0.5) {$U/O$};
  \draw[undiredge] (ln1) -- (le1);
\end{scope}

%%%%%%%%%%%%%%%%%%%% arrow %%%%%%%%%%%%%%%%%%%%
\node at (0.5,0) {$\Rightarrow$};

%%%%%%%%%%%%%%%%%%%% R %%%%%%%%%%%%%%%%%%%%
\begin{scope}[shift={(1,0)}]
  \node at (0,0) {$\emptyset$};
\end{scope}

\end{tikzpicture}%

%% file: diagrams/main-example-trade-rights.tex
\begin{tikzpicture}[StdGraphGrid]

%%%%%%%%%%%%%%%%%%%% L %%%%%%%%%%%%%%%%%%%%
\begin{scope}[shift={(0,0)}, every label/.style={font=\scriptsize}]
  \node[stdnode, label=left:{$1$}] (ln1) at (0,0) {};
  \node[stdnode, label=left:{$2$}] (ln2) at (0,1) {};
  \node[stdnode, label=right:{$3$}] (ln3) at (1,0.5) {};
  \node[henode, label=left:{$4$}] (le1) at (0,-0.5) {$U$};
  \node[henode, label=left:{$5$}] (le2) at (0,0.5) {$U$};
  \node[henode, label=right:{$6$}] (le3) at (1,0) {$O$};
  \node[henode] (le4) at ($(ln1)!0.5!(ln3)$) {$R/W$};
  \draw[undiredge] (ln1) -- (le1);
  \draw[undiredge] (ln2) -- (le2);
  \draw[undiredge] (ln3) -- (le3);
  \draw[graphedge] (ln1) -- (le4) -- (ln3);
\end{scope}

%%%%%%%%%%%%%%%%%%%% arrow %%%%%%%%%%%%%%%%%%%%
\node at (1.5,0.5) {$\Rightarrow$};

%%%%%%%%%%%%%%%%%%%% R %%%%%%%%%%%%%%%%%%%%
\begin{scope}[shift={(2,0)}, every label/.style={font=\scriptsize}]
  \node[stdnode, label=left:{$1$}] (rn1) at (0,0) {};
  \node[stdnode, label=left:{$2$}] (rn2) at (0,1) {};
  \node[stdnode, label=right:{$3$}] (rn3) at (1,0.5) {};
  \node[henode, label=left:{$4$}] (re1) at (0,-0.5) {$U$};
  \node[henode, label=left:{$5$}] (re2) at (0,0.5) {$U$};
  \node[henode, label=right:{$6$}] (re3) at (1,0) {$O$};
  \node[henode] (re4) at ($(rn2)!0.5!(rn3)$) {$R/W$};
  \draw[undiredge] (rn1) -- (re1);
  \draw[undiredge] (rn2) -- (re2);
  \draw[undiredge] (rn3) -- (re3);
  \draw[graphedge] (rn2) -- (re4) -- (rn3);
\end{scope}

\end{tikzpicture}%

%% file: diagrams/main-example-remove-rights.tex
\begin{tikzpicture}[StdGraphGrid]

%%%%%%%%%%%%%%%%%%%% L %%%%%%%%%%%%%%%%%%%%
\begin{scope}[shift={(0,0)}, every label/.style={font=\scriptsize}]
  \node[stdnode, label=left:{$1$}] (ln1) at (0,0) {};
  \node[stdnode, label=right:{$2$}] (ln2) at (1,0) {};
  \node[henode, label=left:{$3$}] (le1) at (0,-0.5) {$U$};
  \node[henode, label=right:{$4$}] (le2) at (1,-0.5) {$O$};
  \node[henode] (le3) at ($(ln1)!0.5!(ln2)$) {$R/W$};
  \draw[undiredge] (ln1) -- (le1);
  \draw[undiredge] (ln2) -- (le2);
  \draw[graphedge] (ln1) -- (le3) -- (ln2);
\end{scope}

%%%%%%%%%%%%%%%%%%%% arrow %%%%%%%%%%%%%%%%%%%%
\node at (1.5,0) {$\Rightarrow$};

%%%%%%%%%%%%%%%%%%%% R %%%%%%%%%%%%%%%%%%%%
\begin{scope}[shift={(2,0)}, every label/.style={font=\scriptsize}]
  \node[stdnode, label=left:{$1$}] (rn1) at (0,0) {};
  \node[stdnode, label=right:{$2$}] (rn2) at (1,0) {};
  \node[henode, label=left:{$3$}] (re1) at (0,-0.5) {$U$};
  \node[henode, label=right:{$4$}] (re2) at (1,-0.5) {$O$};
  \draw[undiredge] (rn1) -- (re1);
  \draw[undiredge] (rn2) -- (re2);
\end{scope}

\end{tikzpicture}%

%% file: diagrams/main-example-get-read-right.tex
\begin{tikzpicture}[StdGraphGrid]

%%%%%%%%%%%%%%%%%%%% L %%%%%%%%%%%%%%%%%%%%
\begin{scope}[shift={(0,0)}, every label/.style={font=\scriptsize}]
  \node[stdnode, label=left:{$1$}] (ln1) at (0,0) {};
  \node[stdnode, label=right:{$2$}] (ln2) at (1,0) {};
  \node[henode, label=left:{$3$}] (le1) at (0,-0.5) {$U$};
  \node[henode, label=right:{$4$}] (le2) at (1,-0.5) {$O$};
  \draw[undiredge] (ln1) -- (le1);
  \draw[undiredge] (ln2) -- (le2);
\end{scope}

%%%%%%%%%%%%%%%%%%%% arrow %%%%%%%%%%%%%%%%%%%%
\node at (1.5,0) {$\Rightarrow$};

%%%%%%%%%%%%%%%%%%%% R %%%%%%%%%%%%%%%%%%%%
\begin{scope}[shift={(2,0)}, every label/.style={font=\scriptsize}]
  \node[stdnode, label=left:{$1$}] (rn1) at (0,0) {};
  \node[stdnode, label=right:{$2$}] (rn2) at (1,0) {};
  \node[henode, label=left:{$3$}] (re1) at (0,-0.5) {$U$};
  \node[henode, label=right:{$4$}] (re2) at (1,-0.5) {$O$};
  \node[henode] (re3) at ($(rn1)!0.5!(rn2)$) {$R$};
  \draw[undiredge] (rn1) -- (re1);
  \draw[undiredge] (rn2) -- (re2);
  \draw[graphedge] (rn1) -- (re3) -- (rn2);
\end{scope}

\end{tikzpicture}%

%% file: diagrams/main-example-degrade-rights.tex
\begin{tikzpicture}[StdGraphGrid]

%%%%%%%%%%%%%%%%%%%% L %%%%%%%%%%%%%%%%%%%%
\begin{scope}[shift={(0,0)}, every label/.style={font=\scriptsize}]
  \node[stdnode, label=left:{$1$}] (ln1) at (0,0) {};
  \node[stdnode, label=right:{$2$}] (ln2) at (1,0) {};
  \node[henode, label=left:{$3$}] (le1) at (0,-0.5) {$U$};
  \node[henode, label=right:{$4$}] (le2) at (1,-0.5) {$O$};
  \node[henode] (le3) at ($(ln1)!0.5!(ln2)$) {$W$};
  \draw[undiredge] (ln1) -- (le1);
  \draw[undiredge] (ln2) -- (le2);
  \draw[graphedge] (ln1) -- (le3) -- (ln2);
\end{scope}

%%%%%%%%%%%%%%%%%%%% arrow %%%%%%%%%%%%%%%%%%%%
\node at (1.5,0) {$\Rightarrow$};

%%%%%%%%%%%%%%%%%%%% R %%%%%%%%%%%%%%%%%%%%
\begin{scope}[shift={(2,0)}, every label/.style={font=\scriptsize}]
  \node[stdnode, label=left:{$1$}] (rn1) at (0,0) {};
  \node[stdnode, label=right:{$2$}] (rn2) at (1,0) {};
  \node[henode, label=left:{$3$}] (re1) at (0,-0.5) {$U$};
  \node[henode, label=right:{$4$}] (re2) at (1,-0.5) {$O$};
  \node[henode] (re3) at ($(rn1)!0.5!(rn2)$) {$R$};
  \draw[undiredge] (rn1) -- (re1);
  \draw[undiredge] (rn2) -- (re2);
  \draw[graphedge] (rn1) -- (re3) -- (rn2);
\end{scope}

\end{tikzpicture}%

%% file: diagrams/main-example-error-graph.tex
\begin{tikzpicture}[StdGraphGrid]

\node[stdnode] (n1) at (0,0) {};
\node[stdnode] (n2) at (0,1) {};
\node[stdnode] (n3) at (1,0.5) {};
\node[henode] (e1) at (0,-0.5) {$U$};
\node[henode] (e2) at (0,0.5) {$U$};
\node[henode] (e3) at (1,0) {$O$};
\node[henode] (e4) at ($(n1)!0.5!(n3)$) {$W$};
\node[henode] (e5) at ($(n2)!0.5!(n3)$) {$W$};

\draw[undiredge] (n1) -- (e1);
\draw[undiredge] (n2) -- (e2);
\draw[undiredge] (n3) -- (e3);
\draw[graphedge] (n1) -- (e4) -- (n3);
\draw[graphedge] (n2) -- (e5) -- (n3);

\end{tikzpicture}%

%% file: diagrams/main-example-rewriting.tex
\begin{tikzpicture}[StdGraphGrid]

%%%%%%%%%%%%%%%%%%%% original graph %%%%%%%%%%%%%%%%%%%%

\begin{scope}[shift={(2.3,0)}]
  \node[stdnode] (n1) at (0,0) {};
  \node[stdnode] (n2) at (0,1) {};
  \node[stdnode] (n3) at (1,0) {};
  \node[stdnode] (n4) at (1,1) {};
  \node[henode] (e1) at (0,-0.5) {$U$};
  \node[henode] (e2) at (0,0.5) {$U$};
  \node[henode] (e3) at (1,-0.5) {$O$};
  \node[henode] (e7) at (1,0.5) {$O$};
  \node[henode] (e4) at ($(n1)!0.5!(n3)$) {$W$};
  \node[henode] (e6) at ($(n2)!0.5!(n4)$) {$W$};
  
  \draw[undiredge] (n1) -- (e1);
  \draw[undiredge] (n2) -- (e2);
  \draw[undiredge] (n3) -- (e3);
  \draw[undiredge] (n4) -- (e7);
  \draw[graphedge] (n1) -- (e4) -- (n3);
  \draw[graphedge] (n2) -- (e6) -- (n4);
\end{scope}

%%%%%%%%%%%%%%%%%%%% rule application 1 %%%%%%%%%%%%%%%%%%%%

\node[label=above:{Rule \ref{fig:main-example-trade-rights}}] at (1.65,0.5) 
{$\Leftarrow$};

\begin{scope}[shift={(4.6,0)}]
  \node[stdnode] (ln1) at (0,0) {};
  \node[stdnode] (ln2) at (0,1) {};
  \node[stdnode] (ln3) at (1,0) {};
  \node[stdnode] (ln4) at (1,1) {};
  \node[henode] (le1) at (0,-0.5) {$U$};
  \node[henode] (le2) at (0,0.5) {$U$};
  \node[henode] (le3) at (1,-0.5) {$O$};
  \node[henode] (le7) at (1,0.5) {$O$};
  \node[henode] (le4) at ($(ln1)!0.5!(ln3)$) {$W$};
  \node[henode] (le5) at ($(ln2)!0.5!(ln3)$) {$R$};
  \node[henode] (le6) at ($(ln2)!0.5!(ln4)$) {$W$};
  
  \draw[undiredge] (ln1) -- (le1);
  \draw[undiredge] (ln2) -- (le2);
  \draw[undiredge] (ln3) -- (le3);
  \draw[undiredge] (ln4) -- (le7);
  \draw[graphedge] (ln1) -- (le4) -- (ln3);
  \draw[graphedge] (ln2) -- (le5) -- (ln3);
  \draw[graphedge] (ln2) -- (le6) -- (ln4);
\end{scope}

%%%%%%%%%%%%%%%%%%%% rule application 2 %%%%%%%%%%%%%%%%%%%%

\node[label=above:{Rule \ref{fig:main-example-get-read-right}}] at (3.95,0.5) 
{$\Rightarrow$};

\begin{scope}
  \node[stdnode] (rn1) at (0,0) {};
  \node[stdnode] (rn2) at (0,1) {};
  \node[stdnode] (rn3) at (1,0) {};
  \node[stdnode] (rn4) at (1,1) {};
  \node[henode] (re1) at (0,-0.5) {$U$};
  \node[henode] (re2) at (0,0.5) {$U$};
  \node[henode] (re3) at (1,-0.5) {$O$};
  \node[henode] (re7) at (1,0.5) {$O$};
  \node[henode] (re5) at ($(rn2)!0.5!(rn3)$) {$W$};
  \node[henode] (re6) at ($(rn2)!0.5!(rn4)$) {$W$};
  
  \draw[undiredge] (rn1) -- (re1);
  \draw[undiredge] (rn2) -- (re2);
  \draw[undiredge] (rn3) -- (re3);
  \draw[undiredge] (rn4) -- (re7);
  \draw[graphedge] (rn2) -- (re5) -- (rn3);
  \draw[graphedge] (rn2) -- (re6) -- (rn4);
\end{scope}

\end{tikzpicture}%

%% file: diagrams/preserving-order-morphisms.tex
\begin{tikzpicture}[x=1.5cm,y=-1.5cm]

%%%%%%%%%%%%%%%%%%%% objects %%%%%%%%%%%%%%%%%%%%
  \node (G0) at (0,0) {$G_0$};
  \node (G1) at (1,0) {$G_1$};
  \node (G2) at (0,1) {$G_2$};
  \node (G3) at (1,1) {$G_3$};

%%%%%%%%%%%%%%%%%%%% arrows %%%%%%%%%%%%%%%%%%%%

\draw[mor-genorder] (G0) -- (G1) node [midway, above] {$\mu$};
\draw[mor-tot] (G0) -- (G2) node [midway, right] {$g$};
\draw[mor-parl] (G1) -- (G3) node [midway, right] {$g'$};
\draw[mor-genorder] (G2) -- (G3) node [midway, above] {$\mu'$};

\end{tikzpicture}%

%% file: diagrams/order-construction-closure.tex
\begin{tikzpicture}[x=1.2cm,y=-1.2cm]

%%%%%%%%%%%%%%%%%%%% left diagram %%%%%%%%%%%%%%%%%%%%
\begin{scope}[shift={(0,0)}]
  \node (La) at (0,0) {$L$};
  \node (Ra) at (1,0) {$R$};
  \node (Ga) at (0,1) {$G$};
  \node (Ha) at (1,1) {$H$};
  \node (Sa) at (2,2) {$S$};
  \draw[mor-parl] (La) -- (Ra) node [midway, above] {$r$};
  \draw[mor-tot] (La) -- (Ga) node [midway, left] {$m$};
  \draw[mor-tot] (Ra) -- (Ha) node [midway, right] {$m'$};
  \draw[mor-parl] (Ga) -- (Ha) node [midway, above] {$r'$};
  \draw[mor-genorder] (Ha) -- (Sa) node [midway, above right] {$\mu$};
\end{scope}

%%%%%%%%%%%%%%%%%%%% right diagram %%%%%%%%%%%%%%%%%%%%
\begin{scope}[shift={(3.5,0)}]
  \node (Lb) at (0,0) {$L$};
  \node (Rb) at (1,0) {$R$};
  \node (Mb) at (2,0) {$R'$};
  \node (Gb) at (0,1) {$G$};
  \node (Hb) at (1,1) {$H$};
  \node (Sb) at (2,2) {$S$};
  \node (Xb) at (0,2) {$G'$};
  \draw[mor-parl] (Lb) -- (Rb) node [midway, above] {$r$};
  \draw[mor-genorder] (Rb) -- (Mb) node [midway, above] {$\mu_R$};
  \draw[mor-tot] (Lb) -- (Gb) node [midway, left] {$m$};
  \draw[mor-tot] (Rb) -- (Hb) node [midway, right] {$m'$};
  \draw[mor-tot] (Mb) -- (Sb) node [midway, right] {$n$};
  \draw[mor-parl] (Gb) -- (Hb) node [midway, above] {$r'$};
  \draw[mor-genorder] (Gb) -- (Xb) node [midway, left] {$\mu_G$};
  \draw[mor-parl] (Xb) -- (Sb) node [midway, above] {$s$};
  \draw[mor-genorder] (Hb) -- (Sb) node [midway, above right] {$\mu$};
\end{scope}

\end{tikzpicture}%

%% file: diagrams/subgraph-example.tex
\begin{tikzpicture}[StdGraphGrid]

%%%%%%%%%%%%%%%%%%%% G1 %%%%%%%%%%%%%%%%%%%%
\begin{scope}[shift={(0,0)}]
  \node[stdnode] (g1n1) at (0,0) {};
  \node[stdnode] (g1n3) at (1,0) {};
  \node[henode] (g1e1) at (0,-0.5) {$U$};
  \node[henode] (g1e3) at (1,-0.5) {$O$};
  \node[henode] (g1e5) at ($(g1n1)!0.5!(g1n3)+(0,-0.25)$) {$W$};
  \node[henode] (g1e6) at ($(g1n1)!0.5!(g1n3)+(0,0.25)$) {$W$};
  \draw[undiredge] (g1n1) -- (g1e1);
  \draw[undiredge] (g1n3) -- (g1e3);
  \draw[graphedge] (g1n1) -- (g1e5) -- (g1n3);
  \draw[graphedge] (g1n1) -- (g1e6) -- (g1n3);
\end{scope}

\node[label=above:{Rule \ref{fig:main-example-add-user}}] at (1.75,0.5) 
{$\Rightarrow$};
%%%%%%%%%%%%%%%%%%%% G2 %%%%%%%%%%%%%%%%%%%%
\begin{scope}[shift={(2.5,0)}]
  \node[stdnode] (g2n1) at (0,0) {};
  \node[stdnode] (g2n2) at (0,1) {};
  \node[stdnode] (g2n3) at (1,0) {};
  \node[henode] (g2e1) at (0,-0.5) {$U$};
  \node[henode] (g2e2) at (0,0.5) {$U$};
  \node[henode] (g2e3) at (1,-0.5) {$O$};
  \node[henode] (g2e5) at ($(g2n1)!0.5!(g2n3)+(0,-0.25)$) {$W$};
  \node[henode] (g2e6) at ($(g2n1)!0.5!(g2n3)+(0,0.25)$) {$W$};
  \draw[undiredge] (g2n1) -- (g2e1);
  \draw[undiredge] (g2n2) -- (g2e2);
  \draw[undiredge] (g2n3) -- (g2e3);
  \draw[graphedge] (g2n1) -- (g2e5) -- (g2n3);
  \draw[graphedge] (g2n1) -- (g2e6) -- (g2n3);
\end{scope}

% \node[label=above:{Rule \ref{fig:main-example-add-object}}] at (4.25,0.5) 
% {$\Rightarrow$};
% %%%%%%%%%%%%%%%%%%%% G3 %%%%%%%%%%%%%%%%%%%%
% \begin{scope}[shift={(5,0)}]
%   \node[stdnode] (g3n1) at (0,0) {};
%   \node[stdnode] (g3n2) at (0,1) {};
%   \node[stdnode] (g3n3) at (1,0) {};
%   \node[stdnode] (g3n4) at (1,1) {};
%   \node[henode] (g3e1) at (0,-0.5) {$U$};
%   \node[henode] (g3e2) at (0,0.5) {$U$};
%   \node[henode] (g3e3) at (1,-0.5) {$O$};
%   \node[henode] (g3e4) at (1,0.5) {$O$};
%   \node[henode] (g3e5) at ($(g3n1)!0.5!(g3n3)+(0,-0.25)$) {$W$};
%   \node[henode] (g3e6) at ($(g3n1)!0.5!(g3n3)+(0,0.25)$) {$W$};
%   \node[henode] (g3e7) at ($(g3n2)!0.5!(g3n4)$) {$W$};
%   \draw[undiredge] (g3n1) -- (g3e1);
%   \draw[undiredge] (g3n2) -- (g3e2);
%   \draw[undiredge] (g3n3) -- (g3e3);
%   \draw[undiredge] (g3n4) -- (g3e4);
%   \draw[graphedge] (g3n1) -- (g3e5) -- (g3n3);
%   \draw[graphedge] (g3n1) -- (g3e6) -- (g3n3);
%   \draw[graphedge] (g3n2) -- (g3e7) -- (g3n4);
% \end{scope}

\node[label=above:{Rule \ref{fig:main-example-trade-rights}}] at (4.25,0.5) 
{$\Rightarrow$};
%%%%%%%%%%%%%%%%%%%% G4 %%%%%%%%%%%%%%%%%%%%
\begin{scope}[shift={(5,0)}]
  \node[stdnode] (g4n1) at (0,0) {};
  \node[stdnode] (g4n2) at (0,1) {};
  \node[stdnode] (g4n3) at (1,0) {};
%  \node[stdnode] (g4n4) at (1,1) {};
  \node[henode] (g4e1) at (0,-0.5) {$U$};
  \node[henode] (g4e2) at (0,0.5) {$U$};
  \node[henode] (g4e3) at (1,-0.5) {$O$};
%  \node[henode] (g4e4) at (1,0.5) {$O$};
  \node[henode] (g4e5) at ($(g4n1)!0.5!(g4n3)$) {$W$};
  \node[henode] (g4e6) at ($(g4n2)!0.5!(g4n3)$) {$W$};
%  \node[henode] (g4e7) at ($(g4n2)!0.5!(g4n4)$) {$W$};
  \draw[undiredge] (g4n1) -- (g4e1);
  \draw[undiredge] (g4n2) -- (g4e2);
  \draw[undiredge] (g4n3) -- (g4e3);
%  \draw[undiredge] (g4n4) -- (g4e4);
  \draw[graphedge] (g4n1) -- (g4e5) -- (g4n3);
  \draw[graphedge] (g4n2) -- (g4e6) -- (g4n3);
%  \draw[graphedge] (g4n2) -- (g4e7) -- (g4n4);
\end{scope}

\end{tikzpicture}%

%% file: diagrams/example2-graph3.tex
\begin{tikzpicture}[StdGraphGrid]

\node[stdnode] (ln1) at (0,0) {};
\node[stdnode] (ln2) at (1,0) {};
\node[henode] (le1) at (0,-0.5) {$U$};
\node[henode] (le2) at (1,-0.5) {$U$};
\node[henode] (le3) at (1.5,0) {$O$};
\node[henode] (le4) at (1,0.7) {$W$};
\node[henode] (le5) at ($(ln1)!0.5!(ln2)$) {$W$};
\draw[undiredge] (ln1) -- (le1);
\draw[undiredge] (ln2) -- (le2);
\draw[undiredge] (ln2) -- (le3);
\draw[undiredge] (ln2) to[bend right] (le4);
\draw[graphedge] (le4) to[bend right] (ln2);
\draw[graphedge] (ln1) -- (le5) -- (ln2);

\end{tikzpicture}%

%% file: diagrams/example2-graph4.tex
\begin{tikzpicture}[StdGraphGrid]

\node[stdnode] (ln1) at (0,0) {};
\node[henode] (le1) at (0,-0.5) {$U$};
\node[henode] (le2) at (0,0.5) {$O$};
\node[henode] (le3) at (0.5,0) {$W$};
\node[henode] (le4) at (-0.5,0) {$W$};
\draw[undiredge] (ln1) -- (le1);
\draw[undiredge] (ln1) -- (le2);
\draw[undiredge] (ln1) to[bend right] (le3);
\draw[graphedge] (le3) to[bend right] (ln1);
\draw[undiredge] (ln1) to[bend right] (le4);
\draw[graphedge] (le4) to[bend right] (ln1);

\end{tikzpicture}%

%% file: diagrams/induced-subgraph-example.tex
\begin{tikzpicture}[StdGraphGrid]

%%%%%%%%%%%%%%%%%%%% L %%%%%%%%%%%%%%%%%%%%
\begin{scope}[shift={(0,0)}, every label/.style={font=\scriptsize}]
  \node[stdnode, label=above:{$1$}] (ln1) at (0,0) {};
  \node[stdnode, label=above:{$2$}] (ln2) at (1,0) {};
  \node[stdnode, label=above:{$3$}] (ln3) at (2,0) {};
  \node[henode, label=above:{$4$}] (le1) at ($(ln1)!0.5!(ln2)$) {$A$};
  \node[henode, label=above:{$5$}] (le2) at ($(ln2)!0.5!(ln3)$) {$A$};
  \node[henode, dashed] (le3) at ($(ln1)!0.5!(ln3)+(0,1)$) {$A$};
  \draw[graphedge] (ln1) -- (le1) -- (ln2);
  \draw[graphedge] (ln2) -- (le2) -- (ln3);
  \draw[graphedge, dashed] (ln1) -- (le3) -- (ln3);
\end{scope}

%%%%%%%%%%%%%%%%%%%% arrow %%%%%%%%%%%%%%%%%%%%
\node at (2.5,0.5) {$\Rightarrow$};

%%%%%%%%%%%%%%%%%%%% R %%%%%%%%%%%%%%%%%%%%
\begin{scope}[shift={(3,0)}, every label/.style={font=\scriptsize}]
  \node[stdnode, label=above:{$1$}] (rn1) at (0,0) {};
  \node[stdnode, label=above:{$2$}] (rn2) at (1,0) {};
  \node[stdnode, label=above:{$3$}] (rn3) at (2,0) {};
  \node[henode, label=above:{$4$}] (re1) at ($(rn1)!0.5!(rn2)$) {$A$};
  \node[henode, label=above:{$5$}] (re2) at ($(rn2)!0.5!(rn3)$) {$A$};
  \node[henode] (re3) at ($(rn1)!0.5!(rn3)+(0,1)$) {$A$};
  \draw[graphedge] (rn1) -- (re1) -- (rn2);
  \draw[graphedge] (rn2) -- (re2) -- (rn3);
  \draw[graphedge] (rn1) -- (re3) -- (rn3);
\end{scope}

\end{tikzpicture}%

%% file: diagrams/pushout.tex
\begin{tikzpicture}[x=1.5cm,y=-1.5cm]

\node (G0) at (0,0) {$G_0$};
\node (G1) at (1,0) {$G_1$};
\node (G2) at (0,1) {$G_2$};
\node (G3) at (1,1) {$G_3$};
\node (G3b) at (2,2) {$G_3'$};
\draw[mor-parl] (G0) -- node [midway, above] {$\phi$} (G1);
\draw[mor-parl] (G0) -- node [midway, left] {$\psi$} (G2);
\draw[mor-parl] (G2) -- node [midway, above] {$\phi'$} (G3);
\draw[mor-parl] (G1) -- node [midway, left] {$\psi'$} (G3);
\draw[mor-parl] (G3) -- node [midway, above right] {$\eta$} (G3b);
\draw[mor-parl] (G1) to[bend left] node [midway, above right] {$\psi''$} (G3b);
\draw[mor-parl] (G2) to[bend right] node [midway, below left] {$\phi''$} (G3b);

\end{tikzpicture}%

%% file: diagrams/proc-correctness1.tex
\begin{tikzpicture}[x=1.5cm, y=-1.5cm]

%%%%%%%%%%%%%%%%%%%% left diagram %%%%%%%%%%%%%%%%%%%%
\begin{scope}[shift={(0,0)}]
  \node (La) at (0,0) {$L$};
  \node (R2a) at (1,0) {$R'$};
  \node (Ga) at (0,1) {$G$};
  \node (Sa) at (1,1) {$S$};
  \draw[mor-parl] (La) -- node[midway, above] {$\mu \circ r$} (R2a);
  \draw[mor-tot] (La) -- node[midway, right] {$m$} (Ga);
  \draw[mor-tot] (R2a) -- node [midway, right] {$m''$} (Sa);
  \draw[mor-parl] (Ga) -- node [midway, above] {$k$} (Sa);
\end{scope}

%%%%%%%%%%%%%%%%%%%% left diagram %%%%%%%%%%%%%%%%%%%%
\begin{scope}[shift={(2.5,0)}]
  \node (Lb) at (0,0) {$L$};
  \node (R1b) at (1,0) {$R$};
  \node (R2b) at (2,0) {$R'$};
  \node (Gb) at (0,1) {$G$};
  \node (S2b) at (1,1) {$S'$};
  \node (S1b) at (2,1) {$S$};
  \draw[mor-parl] (Lb) -- node[midway, above] {$r$} (R1b);
  \draw[mor-genorder] (R1b) -- node[midway, above] {$\mu$} (R2b);
  \draw[mor-tot] (Lb) -- node[midway, right] {$m$} (Gb);
  \draw[mor-tot] (R2b) -- node [midway, right] {$m''$} (S1b);
  \draw[mor-parl] (Gb) -- node [midway, above] {$r'$} (S2b);
  \draw[mor-genorder] (S2b) -- node [midway, above] {$\mu'$} (S1b);
  \draw[mor-tot] (R1b) -- node[midway, right] {$m'$} (S2b);
  \draw[mor-parl, dotted] (Gb) to[bend right] node[midway, below] {$k$} (S1b);
\end{scope}

\end{tikzpicture}%

%% file: diagrams/proc-correctness2.tex
\begin{tikzpicture}[x=1.5cm, y=-1.5cm]

%%%%%%%%%%%%%%%%%%%% left diagram %%%%%%%%%%%%%%%%%%%%
\begin{scope}[shift={(0,0)}]
  \node (G1a) at (-1,1) {$G$};
  \node (La) at (0,0) {$L$};
  \node (R1a) at (1,0) {$R$};
  \node (G2a) at (0,1) {$G_1$};
  \node (G3a) at (1,1) {$G_2$};
  \node (Sa) at (2,2) {$S$};
  \draw[mor-genorder] (G1a) -- node[midway, above] {$\nu$} (G2a);
  \draw[mor-parl] (La) -- node[midway, above] {$r$} (R1a);
  \draw[mor-tot] (La) -- node[midway, right] {$m$} (G2a);
  \draw[mor-tot] (R1a) -- node [midway, right] {$m'$} (G3a);
  \draw[mor-parl] (G2a) -- node [midway, above] {$r'$} (G3a);
  \draw[mor-genorder] (G3a) -- node [midway, above right] {$\mu$} (Sa);
\end{scope}

%%%%%%%%%%%%%%%%%%%% left diagram %%%%%%%%%%%%%%%%%%%%
\begin{scope}[shift={(4,0)}]
  \node (G1b) at (-1,1) {$G$};
  \node (Lb) at (0,0) {$L$};
  \node (R1b) at (1,0) {$R$};
  \node (G2b) at (0,1) {$G_1$};
  \node (G3b) at (1,1) {$G_2$};
  \node (Sb) at (2,2) {$S$};
  \node (R2b) at (2,0) {$R'$};
  \node (G4b) at (0,2) {$G_3$};
  \draw[mor-genorder] (G1b) -- node[midway, above] {$\nu$} (G2b);
  \draw[mor-parl] (Lb) -- node[midway, above] {$r$} (R1b);
  \draw[mor-tot] (Lb) -- node[midway, right] {$m$} (G2b);
  \draw[mor-tot] (R1b) -- node [midway, right] {$m'$} (G3b);
  \draw[mor-parl] (G2b) -- node [midway, above] {$r'$} (G3b);
  \draw[mor-genorder] (G3b) -- node [midway, above right] {$\mu$} (Sb);
  \draw[mor-genorder] (R1b) -- node [midway, above] {$\mu_R$} (R2b);
  \draw[mor-tot] (R2b) -- node [midway, right] {$n$} (Sb);
  \draw[mor-genorder] (G2b) -- node [midway, right] {$\mu_G$} (G4b);
  \draw[mor-parl] (G4b) -- node [midway, above] {$s$} (Sb);
\end{scope}

\end{tikzpicture}%

%% file: diagrams/preserving-subgraph-morphisms.tex
\begin{tikzpicture}[x=1.5cm,y=-1.5cm]

%%%%%%%%%%%%%%%%%%%% objects %%%%%%%%%%%%%%%%%%%%
  \node (G0) at (0,0) {$G_0$};
  \node (G1) at (1,0) {$G_1$};
  \node (G2) at (0,1) {$G_2$};
  \node (G3) at (1,1) {$G_3$};

%%%%%%%%%%%%%%%%%%%% arrows %%%%%%%%%%%%%%%%%%%%

\draw[mor-subgraph] (G0) -- (G1) node [midway, above] {$\mu$};
\draw[mor-tot] (G0) -- (G2) node [midway, right] {$g$};
\draw[mor-parl] (G1) -- (G3) node [midway, right] {$g'$};
\draw[mor-subgraph] (G2) -- (G3) node [midway, above] {$\mu'$};

\end{tikzpicture}%

%% file: diagrams/subgraph-construction-closure.tex
\begin{tikzpicture}[x=1.5cm,y=-1.5cm]

%%%%%%%%%%%%%%%%%%%% left diagram %%%%%%%%%%%%%%%%%%%%
\begin{scope}[shift={(0,0)}]
  \node (La) at (0,0) {$L$};
  \node (Ra) at (1,0) {$R$};
  \node (Ga) at (0,1) {$G$};
  \node (Ha) at (1,1) {$H$};
  \node (Sa) at (2,2) {$S$};
  \draw[mor-parl] (La) -- node [midway, above] {$r$} (Ra);
  \draw[mor-tot] (La) -- node [midway, left] {$m$} (Ga);
  \draw[mor-tot] (Ra) -- node [midway, right] {$m'$} (Ha);
  \draw[mor-parl] (Ga) -- node [midway, above] {$r'$} (Ha);
  \draw[mor-subgraph] (Ha) -- node [midway, above right] {$\mu$} (Sa);
\end{scope}

%%%%%%%%%%%%%%%%%%%% right diagram %%%%%%%%%%%%%%%%%%%%
\begin{scope}[shift={(3.5,0)}]
  \node (Lb) at (0,0) {$L$};
  \node (Rb) at (1,0) {$R$};
  \node (Mb) at (2,0) {$R'$};
  \node (Gb) at (0,1) {$G$};
  \node (Hb) at (1,1) {$H$};
  \node (Sb) at (2,2) {$S$};
  \node (Xb) at (0,2) {$G'$};
  \node (S2b) at (3,3) {$S'$};
  \draw[mor-parl] (Lb) -- node [midway, above] {$r$} (Rb);
  \draw[mor-subgraph] (Rb) -- node [midway, above] {$\mu_R$} (Mb);
  \draw[mor-tot] (Lb) -- node [midway, left] {$m$} (Gb);
  \draw[mor-tot] (Rb) -- node [midway, right] {$m'$} (Hb);
  \draw[mor-tot] (Mb) -- node [midway, right] {$n$} (Sb);
  \draw[mor-parl] (Gb) -- node [midway, above] {$r'$} (Hb);
  \draw[mor-subgraph] (Gb) -- node [midway, left] {$\mu_G$} (Xb);
  \draw[mor-parl] (Xb) -- node [midway, above] {$s$} (Sb);
  \draw[mor-subgraph] (Hb) -- node [midway, above right] {$\mu$} (Sb);
  \draw[mor-parl] (Mb) to[bend left=15] node [midway, above right] 
  {$n'$}(S2b);
  \draw[mor-parl] (Xb) to[bend right=15] node [midway, below left] 
  {$s'$}(S2b);
  \draw[mor-parl] (Sb) -- node [midway, above right] {$\eta'$}(S2b);
  \draw[mor-parl] (Hb) to[bend right=25] node [midway, below left] 
  {$\eta$}(S2b);
\end{scope}

\end{tikzpicture}%

%% file: diagrams/wsts-with-subgraphs.tex
\begin{tikzpicture}[x=1.5cm,y=-1.5cm]

\node (L) at (0,0) {$L$};
\node (R) at (1,0) {$R$};
\node (G) at (0,1) {$G$};
\node (H) at (1,1) {$H$};
\node (G2) at (-1,1.5) {$G'$};
\node (H2) at (2,1.5) {$H'$};
\draw[mor-parl] (L) -- node [midway, above] {$r$} (R);
\draw[mor-tot] (L) -- node [midway, left] {$m$} (G);
\draw[mor-parl] (G) -- node [midway, above] {$r'$} (H);
\draw[mor-tot] (R) -- node [midway, left] {$m'$} (H);
\draw[mor-tot-inj] (G) -- node [midway, above] {$\mu^{-1}$} (G2);
\draw[mor-tot] (L) to[bend right] node [midway, above left] {$m_\mu$} (G2);
\draw[mor-parl] (G2) -- node [midway, above] {$r_\mu$} (H2);
\draw[mor-tot] (R) to[bend left] node [midway, above right] {$m_\mu'$} (H2);
\draw[mor-tot] (H) to node [midway, above] {$\mu'$} (H2);

\end{tikzpicture}%

%% file: diagrams/subgraph-translation.tex
\begin{tikzpicture}[x=1.5cm,y=-1.5cm]

%%%%%%%%%%%%%%%%%%%% left diagram %%%%%%%%%%%%%%%%%%%%
\begin{scope}[shift={(0,0)}]
  \node[stdnode] (hgn1) at (0,0) {};
  \node[stdnode] (hgn2) at (1,0) {};
  \node[stdnode] (hgn3) at (0,1) {};
  \node[stdnode] (hgn4) at (1,1) {};
  \node[stdnode] (hgn5) at (2,0) {};
  \node[henode] (hge1) at ($(hgn1)!0.5!(hgn4)$) {$A$};
  \node[henode] (hge2) at ($(hgn4)!0.5!(hgn5)$) {$B$};
  \draw[undiredge, shorten >=-1pt] (hgn1) -- node[midway, above] {$0$} (hge1);
  \draw[undiredge, shorten >=-1pt] (hgn2) -- node[midway, above] {$1$} (hge1);
  \draw[undiredge, shorten >=-1pt] (hgn3) -- node[midway, below] {$2$} (hge1);
  \draw[undiredge, shorten >=-1pt] (hgn4) -- node[midway, below] {$3$} (hge1);
  \draw[undiredge, shorten >=-1pt] (hgn2) -- node[midway, above] {$0$} (hge2);
  \draw[undiredge, shorten >=-1pt] (hgn5) -- node[midway, above] {$1$} (hge2);
  \draw[undiredge, shorten >=-1pt] (hgn4) -- node[midway, below] {$2$} (hge2);
\end{scope}

%%%%%%%%%%%%%%%%%%%% right diagram %%%%%%%%%%%%%%%%%%%%
\begin{scope}[shift={(4,-0.5)}]
  \node[stdnode, label=above:{$N$}] (gn1) at (0,0) {};
  \node[stdnode, label=above:{$N$}] (gn2) at (2,0) {};
  \node[stdnode, label=above:{$N$}] (gn3) at (0,2) {};
  \node[stdnode, label=above:{$N$}] (gn4) at (2,2) {};
  \node[stdnode, label=above:{$A$}] (gn5e) at (1,1) {};
  \node[stdnode, label=above:{$N$}] (gn6) at (4,0) {};
  \node[stdnode, label=above:{$B$}] (gn7e) at (3,1) {};
  \node[stdnode, label=above:{$0$}] (gn8x) at ($(gn1)!0.5!(gn5e)$) {};
  \node[stdnode, label=above:{$1$}] (gn9x) at ($(gn2)!0.5!(gn5e)$) {};
  \node[stdnode, label=above:{$2$}] (gn10x) at ($(gn3)!0.5!(gn5e)$) {};
  \node[stdnode, label=above:{$3$}] (gn11x) at ($(gn4)!0.5!(gn5e)$) {};
  \node[stdnode, label=above:{$0$}] (gn12x) at ($(gn2)!0.5!(gn7e)$) {};
  \node[stdnode, label=above:{$1$}] (gn13x) at ($(gn6)!0.5!(gn7e)$) {};
  \node[stdnode, label=above:{$2$}] (gn14x) at ($(gn4)!0.5!(gn7e)$) {};
  \draw[undiredge] (gn1) -- (gn8x);
  \draw[undiredge] (gn2) -- (gn9x);
  \draw[undiredge] (gn3) -- (gn10x);
  \draw[undiredge] (gn4) -- (gn11x);
  \draw[undiredge] (gn5e) -- (gn8x);
  \draw[undiredge] (gn5e) -- (gn9x);
  \draw[undiredge] (gn5e) -- (gn10x);
  \draw[undiredge] (gn5e) -- (gn11x);
  \draw[undiredge] (gn2) -- (gn12x);
  \draw[undiredge] (gn6) -- (gn13x);
  \draw[undiredge] (gn4) -- (gn14x);
  \draw[undiredge] (gn7e) -- (gn12x);
  \draw[undiredge] (gn7e) -- (gn13x);
  \draw[undiredge] (gn7e) -- (gn14x);
\end{scope}

\end{tikzpicture}%

%% file: diagrams/subgraph-undecidable-initial.tex
\begin{tikzpicture}[StdGraphGrid]

\node[henode] (e1) at (0,0) {$q_0$};

\begin{scope}[shift={(1,0)}]
  \node[henode] (e2) at (0,0) {$c_1$};
  \node[stdnode] (n1) at (0.5,0) {};
  \draw[undiredge] (n1) -- (e2);
  \node[stdnode] (n1a) at (1.5,1) {};
  \node[stdnode] (n1b) at (1.5,-1) {};
  \node[henode] (e3) at ($(n1)!0.5!(n1a)$) {$X$};
  \node[henode] (e4) at ($(n1)!0.5!(n1b)$) {$X$};
  \draw[graphedge] (n1) -- (e3) -- (n1a);
  \draw[graphedge] (n1) -- (e4) -- (n1b);
  \draw[dotline] (n1a) to[bend right] node[midway, right] {$m$} (n1b);
\end{scope}

\begin{scope}[shift={(3.5,0)}]
  \node[henode] (e5) at (0,0) {$c_2$};
  \node[stdnode] (n2) at (0.5,0) {};
  \draw[undiredge] (n2) -- (e5);
  \node[stdnode] (n2a) at (1.5,1) {};
  \node[stdnode] (n2b) at (1.5,-1) {};
  \node[henode] (e6) at ($(n2)!0.5!(n2a)$) {$X$};
  \node[henode] (e7) at ($(n2)!0.5!(n2b)$) {$X$};
  \draw[graphedge] (n2) -- (e6) -- (n2a);
  \draw[graphedge] (n2) -- (e7) -- (n2b);
  \draw[dotline] (n2a) to[bend right] node[midway, right] {$n$} (n2b);
\end{scope}

\end{tikzpicture}%

%% file: diagrams/subgraph-undecidable-rules.tex
\begin{tikzpicture}[StdGraphGrid]

%%%%%%%%%%%%%%%%%%%% first rule %%%%%%%%%%%%%%%%%%%%
\begin{scope}
  \node[anchor=west] (r1name) at (0,0) {\underline{$(q,\minskyInc{c_1},p)$}:};

  \begin{scope}[shift={(2,-0.5)}]
    \node[henode] (r1le1) at (0,0.5) {$q$};
    \node[henode] (r1le2) at (0.5,0) {$c_i$};
    \node[stdnode] (r1ln1) at (0.5,0.5) {};
    \draw[undiredge] (r1ln1) -- (r1le2);
  \end{scope}
  \node[fit=(r1le1) (r1le2) (r1ln1)] (r1l) {};

  \begin{scope}[shift={(4,-0.5)}]
    \node[henode] (r1re1) at (0,0.5) {$p$};
    \node[henode] (r1re2) at (0.5,0) {$c_i$};
    \node[stdnode] (r1rn1) at (0.5,0.5) {};
    \node[stdnode] (r1rn2) at (1.5,0.5) {};
    \node[henode] (r1re3) at ($(r1rn1)!0.5!(r1rn2)$) {$X$};
    \draw[undiredge] (r1rn1) -- (r1re2);
    \draw[graphedge] (r1rn1) -- (r1re3) -- (r1rn2);
  \end{scope}
  \node[fit=(r1re1) (r1re2) (r1re3) (r1rn1) (r1rn2)] (r1r) {};
  
  \node (r1edge) at ($(r1l.east)!0.5!(r1r.west)$) {$\Rightarrow$};
\end{scope}

%%%%%%%%%%%%%%%%%%%% second rule %%%%%%%%%%%%%%%%%%%%
\begin{scope}[shift={(0,1.25)}]
  \node[anchor=west] (r2name) at (0,0) {\underline{$(q,\minskyDec{c_1},p)$}:};

  \begin{scope}[shift={(2,-0.5)}]
    \node[henode] (r2le1) at (0,0.5) {$q$};
    \node[henode] (r2le2) at (0.5,0) {$c_i$};
    \node[stdnode] (r2ln1) at (0.5,0.5) {};
    \node[stdnode] (r2ln2) at (1.5,0.5) {};
    \node[henode] (r2le3) at ($(r2ln1)!0.5!(r2ln2)$) {$X$};
    \draw[undiredge] (r2ln1) -- (r2le2);
    \draw[graphedge] (r2ln1) -- (r2le3) -- (r2ln2);
  \end{scope}
  \node[fit=(r2le1) (r2le2) (r2le3) (r2ln1) (r2ln2)] (r2l) {};

  \begin{scope}[shift={(5,-0.5)}]
    \node[henode] (r2re1) at (0,0.5) {$p$};
    \node[henode] (r2re2) at (0.5,0) {$c_i$};
    \node[stdnode] (r2rn1) at (0.5,0.5) {};
    \draw[undiredge] (r2rn1) -- (r2re2);
  \end{scope}
  \node[fit=(r2re1) (r2re2) (r2rn1)] (r2r) {};

  \node (r2edge) at ($(r2l.east)!0.5!(r2r.west)$) {$\Rightarrow$};
\end{scope}

%%%%%%%%%%%%%%%%%%%% third rule %%%%%%%%%%%%%%%%%%%%
\begin{scope}[shift={(0,2.5)}]
  \node[anchor=west] (r3name) at (0,0) {\underline{$(q,\minskyZero{c_1},p)$}:};

  \begin{scope}[shift={(2,-0.5)}]
    \node[henode] (r3le1) at (0,0.5) {$q$};
    \node[henode] (r3le2) at (0.5,0) {$c_i$};
    \node[stdnode] (r3ln1) at (0.5,0.5) {};
    \draw[undiredge] (r3ln1) -- (r3le2);
  \end{scope}
  \node[fit=(r3le1) (r3le2) (r3ln1)] (r3l) {};

  \begin{scope}[shift={(4,-0.5)}]
    \node[henode] (r3re1) at (0,0.5) {$p^B$};
    \node[henode] (r3re2) at (0.5,0) {$c_i$};
    \node[stdnode] (r3rn1) at (0.5,0.5) {};
    \draw[undiredge] (r3rn1) -- (r3re2);
    \node[stdnode] (r3rn2) at (1.5,0.5) {};
    \node[henode] (r3re3) at ($(r3rn1)!0.5!(r3rn2)$) {$X$};
    \draw[graphedge] (r3rn1) -- (r3re3) -- (r3rn2);
    \node[stdnode] (r3rn3) at (2.5,0.5) {};
    \node[henode] (r3re4) at ($(r3rn2)!0.5!(r3rn3)$) {$X$};
    \draw[graphedge] (r3rn2) -- (r3re4) -- (r3rn3);
  
  \end{scope}
  \node[fit=(r3re1) (r3re2) (r3re3) (r3re4) (r3rn1) (r3rn2) (r3rn3)] (r3r) {};

  \node (r3edge) at ($(r3l.east)!0.5!(r3r.west)$) {$\Rightarrow$};
\end{scope}

%%%%%%%%%%%%%%%%%%%% auxiliary rule %%%%%%%%%%%%%%%%%%%%
\begin{scope}[shift={(0,3.75)}]
  \node[anchor=west] (r4name) at (0,0) {\underline{$\forall q \in Q$}:};

  \begin{scope}[shift={(2,-0.5)}]
    \node[henode] (r4le1) at (0,0.5) {$q^B$};
    \node[henode] (r4le2) at (0.5,0) {$c_i$};
    \node[stdnode] (r4ln1) at (0.5,0.5) {};
    \draw[undiredge] (r4ln1) -- (r4le2);
    \node[stdnode] (r4ln2) at (1.5,0.5) {};
    \node[henode] (r4le3) at ($(r4ln1)!0.5!(r4ln2)$) {$X$};
    \draw[graphedge] (r4ln1) -- (r4le3) -- (r4ln2);
    \node[stdnode] (r4ln3) at (2.5,0.5) {};
    \node[henode] (r4le4) at ($(r4ln2)!0.5!(r4ln3)$) {$X$};
    \draw[graphedge] (r4ln2) -- (r4le4) -- (r4ln3);
  \end{scope}
  \node[fit=(r4le1) (r4le2) (r4le3) (r4le4) (r4ln1) (r4ln2) (r4ln3)] (r4l) {};

  \begin{scope}[shift={(6,-0.5)}]
    \node[henode] (r4re1) at (0,0.5) {$q$};
    \node[henode] (r4re2) at (0.5,0) {$c_i$};
    \node[stdnode] (r4rn1) at (0.5,0.5) {};
    \draw[undiredge] (r4rn1) -- (r4re2);
  \end{scope}
  \node[fit=(r4re1) (r4re2) (r4rn1)] (r4r) {};

  \node (r4edge) at ($(r4l.east)!0.5!(r4r.west)$) {$\Rightarrow$};
\end{scope}

\end{tikzpicture}%